\documentclass[a4paper,10pt]{article}
\usepackage[toc,page]{appendix}

\newcounter{mnotecount}[section]

\usepackage{ae}

\usepackage[ansinew]{inputenc}

\usepackage{lscape}

\usepackage{amssymb,amsmath,amsthm}

\usepackage{epsf}

\usepackage{soul, xcolor, psfrag}

\usepackage{mathrsfs}

\theoremstyle{definition}

\newtheorem{thm}{Theorem}

\newtheorem{lem}{Lemma}

\newtheorem{prop}{Proposition}

\newtheorem{Remark}{Remark}

\newcommand{\metric}{{\bf g}}

\title{Global solutions to the spherically symmetric Einstein-scalar field system with a positive cosmological constant in Bondi coordinates}

\author{Jo\~ao L. Costa$^{(2,1)}$ and Filipe C. Mena$^{(1,3)}$\\\\
{\small $^{(1)}$Centro de An\'alise Matem\'atica, Geometria e Sistemas Din\^amicos,}
\\
{\small	Instituto Superior T\'ecnico, Universidade de Lisboa, Av. Rovisco Pais 1, 1049-001 Lisboa, Portugal}\\
{\small $^{(2)}$Instituto Universit\'ario de Lisboa (ISCTE-IUL), Av. das For\c{c}as Armadas, 1649-026 Lisboa, Portugal}\\
{\small $^{(3)}$Centro de Matem\'atica, Universidade do Minho, 4710-057 Braga, Portugal}\\
}

\begin{document}

\maketitle

\begin{abstract}
We consider a characteristic initial value problem, with initial data given on a future null cone, for the Einstein (massless) scalar field system with a positive cosmological constant, in Bondi coordinates.
We prove that, for small data,
this system has a unique global classical solution which is causally geodesically complete to the future and decays polynomially in radius and exponentially in Bondi time,  approaching the de Sitter solution.  
\end{abstract}
\section{Introduction}

In this paper we revisit the study of global solutions to the Einstein-scalar field system with a positive cosmological constant $\Lambda>0$, which are triples $(M,\metric,\phi)$, with $(M,\metric)$ a (1+3)-dimensional Lorentzian manifold and $\phi:M\rightarrow\mathbb{R}$ a function, satisfying
\begin{equation}
\label{ELS}
\left\{
\begin{array}{l}
R_{\alpha\beta}-\frac{1}{2}R g_{\alpha\beta}+\Lambda g_{\alpha\beta}=T_{\alpha\beta} \\
T_{\alpha\beta}=\partial_\alpha\phi\partial_\beta \phi-\frac{1}{2} g_{\alpha\beta} g^{\gamma\mu} \partial_\gamma\phi\,\partial_\mu \phi\; \\
\square_{\metric} \phi=0\;,
\end{array}
\right.
\end{equation}
where $R_{\alpha\beta}$, $R$ and $\square_{\metric}$ are, respectively, the Ricci tensor, the scalar curvature and the d'Alembert operator of the metric $\metric$.
This system provides an important toy model for the study of gravitational collapse which retains relevant dynamical degrees of freedom under the assumption of spherical symmetry. The inclusion of the cosmological constant term is motivated by the current standard model of cosmology.

Here we will consider small data solutions under the assumption of Bondi-spherical symmetry, i.e., restricting to metrics of the form
\begin{equation}
\label{metricBondi}
 \metric=-f(u,r)\tilde{f}(u,r)du^{2}-2f(u,r)dudr+r^{2}\sigma_{\mathbb{S}^2}\;,
\end{equation}
where $\sigma_{\mathbb{S}^2}$ is the metric of the unit round two sphere. The system~\eqref{ELS} then becomes
\begin{equation}
\label{BondiSys}
\left\{
\begin{array}{l}
\partial_r f=\frac{1}{2} r f (\partial_r \phi)^2 \\
\partial_r(r\tilde f)=f(1-\Lambda r^2) \\
\left(\partial_u-\frac{1}{2}\tilde f\partial_r\right)\partial_r(r\phi)=\frac{r}{2}\partial_r\tilde f\partial_r\phi\;.
\end{array}
\right.
\end{equation}
Since $u=0$ is a null hypersurface, a natural initial value problem for the previous system is a characteristic initial value problem.
In the asymptotically flat case one has $\Lambda=0$, and the corresponding problem was extensively addressed by Christodoulou (see the introduction in~\cite{Christodoulou-book} for a detailed overview), who initiated the study with the construction of small data dispersive~\footnote{By which we mean future geodesically complete solutions, with empty black hole region, that either approach Minkowski, if $\Lambda=0$, or de Sitter, if $\Lambda>0$.} classical solutions \cite{Christodoulou:1986}. 

Other contributions to the study of dispersive solutions of~\eqref{ELS} in the asymptotically flat case, include the seminal proof of the non-linear stability of Minkowski by Lindblad and Rodnianski~\cite{Lindblad-Rodnianski}, as well as the construction of spherically symmetric dispersive solutions with large BV norms in~\cite{Luk}. In the cosmological setting, Ringstr\"om \cite{Ringstrom} considered the Einstein-scalar field system with a positive potential and proved the exponential decay of (non-linear) perturbations in de Sitter cosmologies. Nonetheless, although far-reaching, those results do not apply to the case of a (non-vanishing) massless scalar field with a positive cosmological constant and, therefore, do not include the case studied in the present paper. 

The study of global properties of solutions to various Einstein-matter equations with a positive cosmological constant has both a long and prestigious tradition as well as a recent remarkable amount of activity, motivated by its rich mathematical structure and relevant physical content. We cannot do justice to the entire literature on the subject here, so we simply refer the interested reader to~\cite{Alho, Andreasson, Chae, dafRen, fajman, Friedrich, Fried-dust, Hadzic, hintzVasy, Hintz, Valiente,  Nungesser, Rad, Todd, Rendall, Rod-Speck, schlue, schlue2, Speck} and references therein.

In~\cite{CostaProblem}, the global existence of solutions to~\eqref{BondiSys} in Bondi time $u$ and their exponentially decay was proved for small initial data. However the results there suffer the undesirable feature of being restricted to a finite radial range.
In fact, by including a positive cosmological constant $\Lambda$ in the Bondi-Christodoulou setting, the apparently simple changes to the PDE system give rise, as expected, to a radical change in the global structure and asymptotic behavior of solutions, which has to be taken into account already at the level of the characteristic initial data imposed on $u=0$. Moreover, further difficulties arise in the analysis of \eqref{BondiSys} as a consequence of $\Lambda$. These complications were latter discussed in~\cite{CostaReview}, where new ideas to overcome them were suggested.  

We will briefly discuss the most relevant of such difficulties: (i) the zeroth order coefficient $\partial_r \tilde f$ decays radially for $\Lambda=0$, but grows linearly for $\Lambda>0$; (ii) the scalar field $\phi$ does not decay to zero when $r\rightarrow \infty$ but instead approaches a limiting function $u\mapsto\underline{\phi}(u)$; (iii) there are characteristics that reach infinity in arbitrarily small Bondi time; (iv) the ``most natural'' iteration schemes for the wave equation do not seem to allow the necessary control over the radial derivatives $\partial^2_r(r\phi)$, which, in turn, are central quantities in the Bondi-Christodoulou setting.

We now revise the strategy developed in this paper to overcome the referred obstacles: First, we carefully choose the radial decay rate of derivatives of initial data, which we take to be of $O(r^{-2+\delta})$, corresponding to a $\delta$ loss of the sharp decay rate of linear waves in de Sitter~\cite{Natario-Sasane}. Then, we derive a new a priori estimate along characteristics that allows us to regain one power of radial decay and deal with (i) and (iii). Afterwards, in view of (ii) and (iv), instead of trying to solve the wave equation directly we first solve the equation obtained by commuting the wave equation with $\partial_r$. This is reminiscent of the linear wave equation in de Sitter where an analogous commutation gives rise to a simpler equation~\cite{CostaSpherically}. Furthermore, it requires the control of an extra radial derivative and, therefore, the need to impose one extra degree of regularity at the level of initial data, which is nevertheless preserved by the evolution. Then, to integrate back to a solution of the wave equation we need to have appropriate boundary data along $r=0$ which we derive by first solving~\eqref{BondiSys} locally, both in Bondi time and radius. This procedure gives rise to solutions which are global in radius but local in Bondi time. Finally, to construct global solutions we use a bootstrap argument based on $L^{\infty}$-energy estimates which also reveal the asymptotic behavior of the solutions. We refer to Theorem~\ref{main-thm} for a precise statement of our main results. 

We finish this introduction by noting that, in principle, there is another natural strategy to solve the system~\eqref{ELS} in spherical symmetry: One first uses Bondi coordinates (see Remark~\ref{remReg}) to solve the system in a region $[0,\infty[\times[0,R]\times \mathbb{S}^2$, with $R$ sufficiently large so that the region includes the entire cosmological horizon. In order to do this, we can simply invoke the results in~\cite{CostaProblem}. Then, we change to appropriate double null coordinates and use the previous (local in radius but global in Bondi time) solution as a source of initial data along the cosmological horizon. In principle, although it remains to be checked, one could adapt the techniques in~\cite{Costa-Natario-Oliveira} to solve the system to the future of the cosmological horizon and, by doing so, we might be able to circumvent some of the previously discussed obstacles. 

Here, however, we rely solely on Bondi coordinates since they provide an unified setting to solve~\eqref{ELS}, they reveal relevant features of the solutions, see for instance the peeling decaying properties in equations~\eqref{phi-convergence}--\eqref{metric-convergence3}, and mostly because we wanted to face the mathematical challenge of overcoming the difficulties posed by the system~\eqref{BondiSys}.                 

\section{Setup and main result}
\label{sectionSetup}

Let $(M, \metric)$ be a spacetime with metric $\metric$ and consider the Einstein equations as
\begin{equation}
\label{EFE}
R_{\alpha\beta}-\frac{1}{2}R g_{\alpha\beta}+\Lambda g_{\alpha\beta}=T_{\alpha\beta},
\end{equation}
where, as usual, $R_{\alpha\beta}$ and $T_{\alpha\beta}$ denote the components of the Ricci tensor and stress energy tensor, respectively, with the greek indices being spacetime indices and $R$ is the Ricci scalar.

We assume a stress energy tensor for a massless scalar field $\phi$  given by
$$
T_{\alpha\beta}=\partial_\alpha\phi\partial_\beta \phi-\frac{1}{2}g_{\alpha\beta} g^{\gamma\mu} \partial_\gamma\phi\,\partial_\mu \phi\;.
$$
Then,
equation \eqref{EFE} reduces to
\begin{equation}
\label{EFE-two}
R_{\alpha\beta}=\partial_\alpha\phi\partial_\beta \phi+\Lambda g_{\alpha\beta}
\end{equation}
and the energy conservation equation $\nabla^\alpha T_{\alpha\beta}=0$ is equivalent to the following wave equation
\begin{equation}
\label{conservation-energy}
\nabla^\alpha\partial_\alpha\phi=0.
\end{equation}
Now, a spacetime is said to be {\em Bondi-spherically symmetric} if it admits a global representation for the metric of the form
\begin{equation}
\label{metricBondi}
 \metric=-f(u,r)\tilde{f}(u,r)du^{2}-2f(u,r)dudr+r^{2}\sigma_{\mathbb{S}^2}\;,
\end{equation}
where $f$ and $\tilde f$ are functions to be determined through the Einstein equations, $\sigma_{\mathbb{S}^2}$ is the round metric of the two-sphere and $r(p):=\sqrt{\text{Area}({\mathcal O}_p)/4\pi}$ the {\em radius function}, where ${\mathcal O}_p$ is the orbit of an $SO(3)$ action by isometries through $p$. The $u$-coordinate is  known as {\em Bondi time} and the future null cones of points at $r=0$ are given by $u=constant$. Note that we have the (gauge) freedom to rescale  the Bondi time $u\mapsto w(u)$ using any increasing and continuously differentiable function $w$.  In view of this freedom we will identify any two Bondi-spherically symmetric spacetimes that differ by such a rescaling.
The coordinates in \eqref{metricBondi} are called {\em Bondi coordinates} and are such that
$$
(u,r)\in [0,U[\times [0,R[, \qquad U,R\in \mathbb{R}^+\cup\{+\infty\}.
$$
Taking \eqref{metricBondi} into \eqref{EFE-two}, we obtain
\begin{equation}
\label{partialfoverf}
\frac{2}{r}\frac{1}{f}\partial_r f=(\partial_r \phi)^2
\end{equation}
as well as
\begin{equation}
\label{drrftilde}
\partial_r(r\tilde f)=f(1-\Lambda r^2),
\end{equation}
while \eqref{conservation-energy} gives
\begin{equation}
\label{new-last-eq}
\frac{
1}{r}\left(\partial_u-\frac{\tilde f}{2}\partial_r\right)\partial_r(r\phi)=\frac{1}{2}(\partial_r\tilde f)(\partial_r\phi).
\end{equation}
As a special case, we recall that the causal future of any point in de Sitter spacetime may be covered by Bondi coordinates with $f^{\mathrm {dS}}(u,r)=1$ and $\tilde f^{\mathrm {dS}}(u,r)=1-\Lambda r^{2}/3$, i.e. with the metric
\begin{equation}
\metric^{\mathrm {dS}}=-\left(1-\frac{\Lambda}{3}r^{2}\right)du^{2}-2dudr+r^{2}\sigma_{\mathbb{S}^2}\;.
\label{dSBondi}
\end{equation}
We now summarise the main result of this paper whose proof is in Section \ref{final-sec}:
\begin{thm}
\label{main-thm}
Let $\Lambda>0$, $0<\delta<1/2$, $k\in\mathbb{Z}^+$ and let $\phi_0\in C^{k+2}([0,+\infty[)$ be  such that
$$\sup_{r\geq 0} \left( |\partial_r(r\phi_0)(r)|+|(1+r)^{2-\delta}\partial^2_r(r\phi_0)(r)| + |(1+r)^{3-\delta}\partial^3_r(r\phi_0)(r)| \right)<\infty\;.$$
There exists $\varepsilon_0=\varepsilon_0(\Lambda,\delta)>0$  such that, if
\begin{equation}
\label{smallness}
\sup_{r\geq 0 } |(1+r)^{2-\delta}\partial^2_r(r\phi_0)(r)|< \varepsilon_0 \;,
\end{equation}
then there exists a unique solution $(M=\mathbb{R}^+_0\times\mathbb{R}^+_0\times \mathbb{S}^2, \metric,\phi)$ for the system \eqref{EFE-two}--\eqref{conservation-energy}, with $(M,\metric)$ a $C^{k+1}$ Bondi-spherically symmetric spacetime and $\phi\in C^{k+2}(M)$ satisfying the initial condition
$$
\phi (0,r,\omega)=\phi_0(r), \text{ for all } r\geq 0 \text{ and } \omega\in\mathbb{S}^2\;.
$$
Moreover:
\begin{enumerate}
\item If we set the Bondi time to be the proper time of the observer at the center of symmetry, i.e., if we set   $f(u,r=0)\equiv 1$, then there exists a continuous function $\underline \phi:[0,+\infty[ \rightarrow \mathbb{R}$ such that
\begin{equation}
\label{phi-convergence}
|\phi(u,r,\omega)-\underline \phi(u) | \lesssim \frac{1}{(1+r)^{1-\delta}} e^{-(1+\delta/2)\sqrt {\Lambda/3}\,u}\;,
\end{equation}
\begin{equation}
\label{phi-convergence2}
|\partial_r\phi(u,r,\omega) | \lesssim \frac{1}{(1+r)^{2-\delta}} e^{-(1+\delta/2)\sqrt {\Lambda/3}\,u}\;,
\end{equation}
\begin{equation}
\label{phi-convergence2}
|\partial^2_r\phi(u,r,\omega) | \lesssim \frac{1}{(1+r)^{3-\delta}} e^{-(1+\delta/2)\sqrt {\Lambda/3}\,u}\;.
\end{equation}
Also, there exists $\underline{\phi}(\infty)\in\mathbb{R}$ such that, given $R>0$, if $\varepsilon_0 \leq\underline{\varepsilon}(R)$, with the later sufficiently small, then there exists a constant  $C_R>0$ such that 
\begin{equation}
\label{phi-convergenceR}
\sup_{r\leq R}\left(|\phi(u,r,\omega)-\underline \phi(\infty) |+|\partial_r\phi(u,r,\omega) |+|\partial^2_r\phi(u,r,\omega) | \right)\leq C_R e^{-2\sqrt {\Lambda/3}\,u}\;.
\end{equation}
Concerning the metric we have 
\begin{equation}
\label{metric-convergenceR}
\sup_{r\leq R}{|f(u,r) - 1 |} \leq C_R  e^{-  4\sqrt {\Lambda/3}\,u}\;, 
\end{equation}
and
\begin{equation}
\label{metric-convergenceR2}
\sup_{r\leq R}{|\tilde{ f}(u,r) - (1-\frac{\Lambda}{3}r^2) |} \leq C_R  e^{-  4\sqrt {\Lambda/3}\,u}\;, 
\end{equation}
where $f$ and $\tilde{f}$ determine, according to~\eqref{metricBondi}, the spacetime metric components of the solution in the $(u,r,\omega)$ Bondi coordinates  (recall that $f^{\mathrm {dS}}= 1$ and $\tilde{f}^{\mathrm {dS}}=1-\frac{\Lambda}{3}r^2$). 
\item Let $R\gg 1$. Fix Bondi time by imposing $d\hat u=f(u,r=\infty) du$.  
Let $(e_{I})_{I=0,1,2,3}$ be an orthonormal frame in the $\{r>R\}$ region of de Sitter spacetime. There exists a diffeomorphism, mapping this region of de Sitter to the region $\{r>R\}$ in our spacetime such that, by writing $g_{IJ}(\hat u,r)=g(e_I,e_J)_{(\hat u, r, \omega)}$, we have
\begin{equation}
\label{metric-convergence3}
|g_{IJ}(\hat u,r)-g^{\mathrm {dS}}_{IJ}(\hat u,r)|\lesssim \frac{1}{(1+r)^{2(1-\delta)}} e^{-  2(1-\varepsilon)(1+\delta/2)\sqrt {\Lambda/3}\,\hat u}\;,
\end{equation}
where $\varepsilon>0$ is a constant that can be made arbitrarily small by decreasing $\varepsilon_0$. 
\item $(M,\metric)$ is causally geodesically complete towards the future.
\end{enumerate}
\end{thm}
\begin{Remark}
\label{remReg}
The wave equation \eqref{conservation-energy} implies that $\phi$ is regular at the centre of symmetry. For instance,  to obtain a $C^1$ scalar field at the center we need the vanishing of the derivatives along the orthogonal space to $r=0$. This, in turn, is equivalent to $\partial_u\phi(u,0)=\tilde f\partial_r\phi(u,0)$ which is  an immediate consequence of the wave equation.
This regularity ``for free'' is a remarkable advantage of using Bondi coordinates.
\end{Remark}
\begin{Remark} 
\label{rmkNonConst}
Note that the function $\underline{\phi}$ clearly does not have to be a constant. As an example consider initial data satisfying $\phi_0(r)=1$, for $r\leq R$, with $R\gg1$, and $\phi_0(r)\rightarrow 0$, $r\rightarrow \infty$. Then $\phi\equiv 1$ in the future domain of dependence of $\{u=0,r\leq R\}$ and, consequently, there exists $\underline u>0$ such that $\underline{\phi}(u)=1$, for all $u\geq \underline u$, and $\underline{\phi}(0)=0$. Moreover, by continuity, $\underline{\phi}$ takes all the values in the interval $[0,1]$.
\end{Remark}
\begin{Remark} 
We note that the constructed spacetimes are {\em weakly asymptotically simple}, they can be conformally compactified with conformal factor $\Omega=1/r$ and the resulting future infinities $\mathscr{I}^+$ are spacelike. Furthermore, for each spacetime, there exists a cosmological   horizon (and a cosmological apparent horizon) relative to $r=0$ which is complete to the future~\cite{CostaProblem}[Section 3].
\end{Remark}

\section{Einstein-$\Lambda$-scalar field system in Christodoulou's variables}

From now on, to simplify the notation, we fix $\Lambda=3$. This possibility is allowed since the Einstein equations can be rescaled with $\Lambda$.

For Bondi-spherically symmetric spacetimes, if one introduces the quantity
\begin{equation}
\label{def-h}
h={\partial}_r(r\phi),
\end{equation}
the full content of the Einstein-$\Lambda$-scalar field equations~\eqref{EFE-two}--\eqref{conservation-energy} is encoded in the integro-differential scalar equation (see also \cite{Christodoulou:1986})
\begin{equation}
\label{mainEq}
Dh=G\left(h-\bar{h}\right)\;,
\end{equation}
derived from~\eqref{new-last-eq},
where the differential operator (``incoming'' null vector field) is given by
\begin{equation}
\label{D}
 D=\partial_ u-\frac{1}{2}\tilde{f}\,\partial_r\;,
\end{equation}
while
\begin{equation}
\label{G}
 G:=\frac{1}{2}\partial_r\tilde f=\frac{1}{2r}\left[(f-\tilde{f})-3 f r^{2}\right]\;,
\end{equation}
and the radial average of a function is given by
\begin{equation}
\label{scalar}
{\bar h}(u,r):=\frac{1}{r}\int_0^rh(u,s)ds\;.
\end{equation}
It is useful to note that the scalar field $\phi$ is recovered from $h$ by averaging as
\begin{equation}
\label{scalar}
\phi={\bar h}\;
\end{equation}
and that
\begin{equation}\label{drbar}
\partial_r\bar h=\frac{h-\bar h}{r}\;.
\end{equation}
By setting the condition
\begin{equation}
\label{boundary-cond}
f(u,r=0)=1,
\end{equation}
which can always be done by an appropriate rescaling
of the $u$ coordinate,
the metric coefficients are obtained from $h$ by the relations:
\begin{equation}
\label{metricQuoficients}
f(u,r)=\exp\left({\frac{1}{2}\int_{0}^{r}\frac{\left(h(u,s)-\bar{h}(u,s)\right)^{2}}{s}ds}\right) \quad,\quad\tilde{f}(u,r)=\frac{1}{r}\int^{r}_{0}(1-3\, s^2)f(u,s)ds\;.
\end{equation}
We will also need an evolution equation for $\partial_rh$  which, by differentiating~\eqref{mainEq} given a sufficiently regular solution, turns out to be
\begin{equation}
\label{D_partial_h}
 (D-2G)\partial_{r}h=-J\,\partial_{r}\bar{h}\;,
\end{equation}
for
\begin{equation}
\label{defJ}
J:=G-r\partial_rG= 3G+3 r f -\frac{1-3 r^{2}}{2}\partial_r f\;.
\end{equation}
\begin{Remark}
	\label{remark-final-sol}
	We remark the fact that if $\phi$ and $\metric$, given by \eqref{metricBondi}, satisfy \eqref{EFE-two} and \eqref{conservation-energy} (or equivalently \eqref{partialfoverf}-\eqref{new-last-eq}), then $h$ defined by \eqref{def-h} solves \eqref{mainEq}. Conversely, if $h$ is a sufficiently regular solution to \eqref{mainEq}, then $\phi=\bar h$, $f$ and $\tilde f$ given by \eqref{metricQuoficients}, solve \eqref{EFE-two}-\eqref{conservation-energy}.
\end{Remark}
\section{Characteristics and a priori estimates}

\newcommand{\Np}{L^{\infty,p}_r}

\newcommand{\Linfty}{L^{\infty}}

\newcommand{\Niid}{L^{\infty,2-\delta}}

\newcommand{\Nid}{L^{\infty,1-\delta}}

\newcommand{\Niiid}{L^{\infty,3-\delta}}

\newcommand{\Nivd}{L^{\infty,3-\delta}}

\newcommand{\yup}{L^{\infty}_UL^{\infty,p}_r}

\newcommand{\yuo}{L_U^{\infty}L^{\infty}_r}

\newcommand{\yupI}{Y_U^{1-\delta}}

\newcommand{\yupII}{L_U^{\infty}L^{\infty,2-\delta}_r}

\newcommand{\yupIII}{L_U^{\infty}L^{\infty,3-\delta}_r}

\newcommand{\yupIV}{L_U^{\infty}L^{\infty,4-\delta}_r}

\newcommand{\xup}{X_U^p}

\newcommand{\xupII}{X_U^{2-\delta}}

\newcommand{\Ndrh}{\|\partial_r h(u,\,\cdot\,)\|_{\Niid}}

\newcommand{\Ydrh}{\|\partial_r h\|_{\yupII}}

\newcommand{\Cd}{C_{\delta}}
%

\subsection{Norms and basic estimates}

Consider $U\in\; ]0,+\infty]$, $p\in \mathbb{R}$ and $f:[0,U]\times[0,+\infty[\rightarrow\mathbb{R}$ a continuous function. We define
$$\|f(u,\,\cdot\,)\|_{\Np}:=\sup_{r\geq0}\left|(1+r)^p f(u,r)\right|\;,$$
and for $p=0$ we will simply write $\Linfty_r=L^{\infty,0}_r$. We will also consider the norms
$$\|f\|_{\yup}:=\sup_{0\leq u\leq U}\|f(u,\,\cdot\,)\|_{\Np}\;.$$
Now, since
\begin{equation}
\label{hhbar1}
\|\bar h(u,\,\cdot\,)\|_{\Linfty_r}\leq \|h(u,\,\cdot\,)\|_{\Linfty_r}
\end{equation}
then
\begin{equation}
\label{hhbar2}
\|\bar h\|_{\yuo}\leq \|h\|_{\yuo}\;.
\end{equation}
Note also that, for any appropriately regular function $w$, we have
\begin{eqnarray*}
|w(u,r)-\bar w(u,r)|
&=&
\left|\frac{1}{r}\int_0^r w(u,r)-w(u,s) ds\right|
\\
&=&
\left|\frac{1}{r}\int_0^r\int_s^r \partial_r w(u,\rho)d\rho \, ds\right|
\\
&\leq&
\frac{1}{r}\int_0^r\left(\int_s^r \frac{1}{(1+\rho)^{p}}\|\partial_r w(u,\,\cdot\,)\|_{L^{\infty,p}_r}\,d\rho\right) ds\;.
\end{eqnarray*}
Then, for $p=0$ we get
\begin{equation}
\label{hhbar0}
|w(u,r)-\bar w(u,r)|\leq \frac{r}{2} \|\partial_r w(u,\,\cdot\,)\|_{\Linfty_r} \;,
\end{equation}
while for $p=2-\delta$, with $0<\delta<1$, this leads to
$$|w(u,r)-\bar w(u,r)|\leq C_{\delta}\|\partial_r w(u,\,\cdot\,)\|_{\Niid_r} \frac{r}{(1+r)^{2-\delta}} W(r)\;,$$
where $C_{\delta}>0$ is a constant depending on $\delta$ and 
$$W(r)=\frac{(1+r)\left[1+(1-\delta)r-(1+r)^{1-\delta}\right]}{r^2}\;.$$
Since $W$ has a finite limit as $r\rightarrow \infty$, we can combine the last inequality with~\eqref{hhbar0} to obtain
\begin{equation}
\label{hhbar}
|w(u,r)-\bar w(u,r)|\leq C_{\delta}\frac{r}{(1+r)^{2-\delta}}\|\partial_r w(u,\,\cdot\,)\|_{\Niid_r}\;.
\end{equation}
A similar reasoning for $p>2$ provides the estimate
\begin{equation}
\label{hhbarp}
|w(u,r)-\bar w(u,r)|\leq C_{p}\frac{r}{(1+r)^{2}}\|\partial_r w(u,\,\cdot\,)\|_{\Np}\;,
\end{equation}
while for $p=2$ we get
\begin{equation}
\label{hhbarp2}
|w(u,r)-\bar w(u,r)|\leq C_{2}\frac{r\log{(2+r)}}{(1+r)^2}\|\partial_r w(u,\,\cdot\,)\|_{L_{r}^{\infty,2}}\;.
\end{equation} 
It is of interest to note that, in the previous estimates, the lost of radial decay  is the smallest in the case $p=2-\delta$.

Using~\eqref{hhbar} we conclude that
\begin{eqnarray}
\int_0^r\frac{\left(h(u,s)-\bar h(u,s)\right)^2}{s}\,ds
\leq
\label{ineq-for-lem}
C_{\delta}\|\partial_r h(u,\,\cdot\,)\|_{\Niid_r}^2\;,
\end{eqnarray}
which, using  \eqref{metricQuoficients}, immediately leads to the following estimate that measures the deviation from the de Sitter solution:
\begin{equation*}
1\leq f(u,r)\leq \exp\left(\Cd\|\partial_r h\|_{\yupII}^2\right)\;\; ,\; \forall (u,r)\in[0,U]\times [0,\infty[\;.
\end{equation*}
We can rewrite the previous inequality as
\begin{equation}
\label{fEstimate}
1\leq f(u,r)\leq 1+\varepsilon\;\; ,\; \forall (u,r)\in[0,U]\times [0,\infty[\;,
\end{equation}
where, from now on, $\varepsilon>0$ will represent a quantity that can be made arbitrarily small by decreasing
$\|\partial_r h\|_{\yupII}^2$.

As a consequence, arguing as in Section 4.1 of~\cite{CostaProblem}, we obtain
\begin{equation}
\label{ftildeabove}
1-(1+\varepsilon)r^2\le \tilde f(u,r)\leq 1+\varepsilon-r^2 \;,\;\; \forall (u,r)\in[0,U]\times [0,\infty[\;,
\end{equation}
 and these inequalities can then be used to estimate the characteristics of equation~\eqref{mainEq}, which are the integral curves of the operator $D$. The characteristic through $(u_1,r_1)$  will be denoted by
\begin{equation}
\chi(u)=\chi(u;u_1,r_1)=(u,r(u;u_1,r_1))\;
\end{equation}
and its radial component is the unique solution of the ordinary differential equation
\begin{equation}
 \frac{dr}{du}=-\frac{1}{2}\tilde{f}(u,r)\;,
\label{Characteristic_ODE}
\end{equation}
satisfying
$r(u_1)=r_1\;.$
Explicit estimates for these characteristics, valid for appropriately small $\|\partial_r h\|_{\yupII}^2$, can be found in~\cite[(30)-(33)]{CostaProblem}.

Now, since $\partial_r f\geq0$ then
$$\bar f(u,r)=\frac{1}{r}\int_0^r f(u,s)ds\leq f(u,r)\;.$$
Using the definition \eqref{metricQuoficients}, the inequality \eqref{hhbar}, together with $f\leq 2$, for small enough $\|\partial_r h\|_{\yupII}^2$, gives
\begin{equation}
\label{drfEst}
0\leq \partial_r f = \frac{1}{2} f \frac{\left(h-\bar h\right)^2}{r}\leq C_{\delta} \frac{r}{(1+r)^{4-2\delta}}\|\partial_r h\|_{\yupII}^2\;,
\end{equation}
which, in view of~\eqref{hhbarp}, leads to
\begin{equation}
\label{fbarfrough}
0\leq f-\bar f \leq C_{\delta}\frac{r}{(1+r)^{2}}\|\partial_r h\|_{\yupII}^2\;,
\end{equation}
provided $0<\delta<1/2$.

We are now able to obtain appropriate estimates for important coefficients of our evolution equations:

\begin{lem}
	\label{lem1}
	Let $0<\delta<1/2$ and $U\in\; ]0,+\infty]$. There exists $x=x(\delta)>0$ such that, if $\|\partial_r h\|_{\yupII}<x$, then the following estimates hold, in $[0,U]\times[0,\infty[$,
	\begin{equation}
	\label{Gestimate}
	G(u,r)\leq-(1-\varepsilon)\,r\;,
	\end{equation}
	\begin{equation}
	\label{absGestimate}
	|G(u,r)|\leq(1+\varepsilon)\,r\;,
	\end{equation}
	and
	\begin{equation}
	\label{Jestimate}
	|J(u,r)|\leq C_{\delta}\frac{1}{(1+r)^{1-2\delta}}\|\partial_r h\|_{\yupII}^2\;,
	\end{equation}
	where $\varepsilon>0$ represents a quantity that vanishes when $x\rightarrow 0$.
	\end{lem}
\begin{proof}
The proof of estimates~\eqref{Gestimate} and~\eqref{absGestimate} is analogous to the proof of Lemma 1 of~\cite{CostaProblem}.
To establish the final estimate note that, since $f$ is increasing in $r$,
\begin{eqnarray*}
\left|\frac{3}{r^2} \int_0^r s^2 f(u,s)ds -r f(u,r)\right| =  \left |\frac{3}{r^2} \int_0^r s^2(f(u,s)-f(u,r)) ds \right| 
 \leq  C ,
\end{eqnarray*}
where $C>0$ is constant and the last inequality follows from the boundedness of $f$.

But we also have
\begin{eqnarray*}
f(u,r)-f(u,s) &=&  \partial_r f(u,\rho)  (r-s) \;\;, \text{ with }  s<\rho<r\;, \\
 &\lesssim  &
 \frac{\|\partial_r h\|_{\yupII}^2}{(1+\rho)^{3-2\delta}} r \\
 &\leq  &
 \frac{\|\partial_r h\|_{\yupII}^2}{(1+s)^{3-2\delta}} r \;,
\end{eqnarray*}
where $A\lesssim B$ means that $A\leq C_{\delta} B$, which gives
\begin{eqnarray*}
\left|\frac{3}{r^2} \int_0^r s^2 f(u,s)ds -r f(u,r)\right| \lesssim \frac{(1+r)^{2\delta}}{r} \|\partial_r h\|_{\yupII}^2\;.
\end{eqnarray*}
Since we already saw that this quantity is also bounded, we conclude that
\begin{eqnarray}
\label{JpartEst}
\left|\frac{3}{r^2} \int_0^r s^2 f(u,s)ds -r f(u,r)\right| \lesssim \frac{1}{(1+r)^{1-2\delta}} \|\partial_r h\|_{\yupII}^2\;.
\end{eqnarray}
Now, using~\eqref{defJ}, ~\eqref{G} and~\eqref{metricQuoficients} we can write
$$|J|\leq\frac{3}{2}\left|\frac{3}{r^2} \int_0^r s^2 f(u,s)ds -r f(u,r)\right|+ 3 \frac{|f-\bar f|}{2r}+\frac{|1-3r^2|}{2}\partial_r f$$
and~\eqref{Jestimate} follows from~\eqref{JpartEst},~\eqref{fbarfrough} and~\eqref{drfEst}.

\end{proof}
\subsection{A priori estimates along characteristics}
\label{sectionAPriori}

We now establish estimates for relevant quantities integrated along characteristics. In particular we will obtain one of the main ingredients missing in~\cite{CostaProblem} which will allow us, later on, to recover one power of decay in the radial direction. In turn, this will be essential to close our estimates in the entire radial range $r\geq 0$.

%
\begin{lem}
\label{main}
Let  $m>0$ and $0<\delta<1/2$.
Given $\varepsilon>0$, there exists $x,C>0$ such that, if $\|\partial_r h\|_{\yupII}<x$, then, for $u\leq u_1\leq U$,
\begin{equation}
\label{expG}
e^{m\int_u^{u_1}G(s,r(s;u_1,r_1))ds}\leq C \left(\frac{1+r(u;u_1,r_1)}{1+r_1}\right)^{2m-\varepsilon}\;.
\end{equation}
%
Moreover, for $p\in\mathbb{R}$ such that $2m>p+1$, there exists $x,C>0$ such that, if $\|\partial_r h\|_{\yupII}<x$, then, for $u_1\leq U$,
\begin{equation}
\label{main-estimate}
\int_0^{u_1} \frac{1}{\left(1+r(u;u_1,r_1)\right)^p}\,e^{m\int_u^{u_1}G(s,r(s;u_1,r_1))ds}du\leq
\frac{C(1+u_1)}{(1+r_1)^{p+1}}\;.
\end{equation}
\end{lem}

\begin{proof}
Let $R>0$. If $r_1\leq R$, we can choose $x$ appropriately small so that~\eqref{Gestimate} holds, and we immediately get the desired result since
$$\int_{0}^{u_1} \frac{1}{\left(1+r(u)\right)^p}\,e^{m\int_u^{u_1}G(s,r(s))ds}du\leq u_1 \leq \frac{(1+R)^{p+1}}{(1+r_1)^{p+1}} (1+u_1)\;.$$
Now consider the case $r_1>R$  and let $r(u)=r(u;u_1,r_1)$.	
We start by observing that for $x$ sufficiently small and $R$ sufficiently large, according to~\eqref{ftildeabove}, we have
\begin{equation}
\label{ftildeLargeR}
 (1-\varepsilon_{x,R}) r^2\leq-\tilde f(u,r)\leq (1+\varepsilon_{x,R}) r^2 \;, \text{ for all } r>R\;,
\end{equation}
where, from now on,  $\varepsilon_{x,R}>0$ represents a quantity that vanishes when both $x\to 0$ and $R\to\infty$. Note that this last estimate also shows that all characteristics that start with $r_1$ sufficiently large are radially increasing.

Then, if we define
$$u_R=\max\left(\{u\in[0,u_1] :r(u)=R\}\cup \{0\}\right)\;,$$
we see that
$$ \frac{2(1-\varepsilon_{x,R})}{\left(r(u)\right)^2}\frac{dr}{du}\leq 1=\frac{2}{-\tilde f(u,r(u))}\frac{dr}{du}\leq \frac{2(1+\varepsilon_{x,R})}{\left(r(u)\right)^2}\frac{dr}{du}\;,\text{ for all }u\in[u_R,u_1]\;.$$
The previous inequalities together with~\eqref{Gestimate} allow us to estimate, for $u_R\leq u \leq u_1$,
\begin{eqnarray}
\nonumber
\int_u^{u_1}G(s,r(s))ds
&\leq&
-(1-\varepsilon)\int_{u}^{u_1}{r(s)}ds
\\
\nonumber
&\leq&
-2(1-\varepsilon_{x,R})\int_{u}^{u_1} \frac{1}{r(s)}\frac{dr}{ds}ds
\\
&=&
\nonumber
-2(1-\varepsilon_{x,R})\log \left(\frac{r_1}{r(u)}\right)\;,
\end{eqnarray}
from which we conclude that
\begin{equation}
\label{eG0}
e^{m\int_u^{u_1}G(s,r(s))ds}\le \left(\frac{r(u)}{r_1}\right)^{2(1-\varepsilon_{x,R})m}\;.
\end{equation}
By hypothesis $2m>p+1$ and we can increase $R$ and decrease $x$, if necessary, to make sure that $2(1-\varepsilon_{x,R})m-p-1>0$. Then, using the basic fact that
$r_1>(1-\varepsilon_{x,R})(1+r_1)\;,$
we obtain
\begin{eqnarray}
\nonumber
\int_{u_R}^{u_1} \frac{1}{\left(1+r(u)\right)^p}\,e^{m\int_u^{u_1}G(s,r(s))ds}du
&\leq&
\int_{u_R}^{u_1} \frac{2(1+\varepsilon_{x,R})}{\left(1+r(u)\right)^p} \left(\frac{r(u)}{r_1}\right)^{2(1-\varepsilon_{x,R})m}\frac{1}{(r(u))^2}\frac{dr}{du}du
\\
\nonumber
&\leq&
\frac{2(1+\varepsilon_{x,R})}{(1+r_1)^{2(1-\varepsilon_{x,R})m}}\int_{r(u_R)}^{r_1} (1+r)^{2(1-\varepsilon_{x,R})m-p-2} dr
\\
\label{mainEstimate1}
&\leq&\frac{2(1+\varepsilon_{x,R})}{2(1-\varepsilon_{x,R})m-(p+1)} \frac{1}{(1+r_1)^{p+1}}\;.
\end{eqnarray}
To estimate the contribution to~\eqref{main-estimate} in the interval $[0,u_R]$ we use~\cite[(32)]{CostaProblem} as
\begin{equation}
\label{estCharExplicit}
r(u)\geq (1-\varepsilon)\coth{\left\{\frac{1+\varepsilon}{2}(c^{-}-u)\right\}}\;,\forall u\leq u_1\;,
\end{equation}
which is valid for $r_1>R$, with $R$ chosen sufficiently large, and where $c^-$ is an integration constant that can be derived by recalling that, in the previous estimate, equality is attained at $u=u_1$.
Then
\begin{eqnarray*}
\int_u^{u_1} r(v)dv
&\geq&
-2\int_u^{u_1}\left(-\frac{1+\varepsilon}{2}\,\frac{\cosh{\left[\frac{1+\varepsilon}{2}(c^{-}-v)\right]}}{\sinh{\left[\frac{1+\varepsilon}{2}(c^{-}-v)\right]}}\right)dv
\\
&=&
-2\log \left(\frac{\sinh{\left[\frac{1+\varepsilon}{2}(c^{-}-u_1)\right]}}{\sinh{\left[\frac{1+\varepsilon}{2}(c^{-}-u)\right]}}\right)\;,
\end{eqnarray*}
and, from \eqref{Gestimate}
\begin{eqnarray}
\nonumber
e^{m\int_u^{u_1}G(v,r(v))dv}
&\leq&
e^{-(1-\varepsilon)m\int_u^{u_1}r(v)dv}
\\
\label{intExp}
&\leq&
\left(\frac{\sinh{\left[\frac{1+\varepsilon}{2}(c^{-}-u_1)\right]}}{\sinh{\left[\frac{1+\varepsilon}{2}(c^{-}-u)\right]}}\right)^{2m(1-\varepsilon)}
\;.
\end{eqnarray}
From the elementary identity
$\sinh x=1/\sqrt{\coth^2x-1}\;,$
using the equality in~\eqref{estCharExplicit}, we get
$$\sinh\left[\frac{1+\varepsilon}{2}(c^--u_1)\right]=\frac{1}{\sqrt{\coth^2[\frac{1+\varepsilon}{2}(c^--u_1)]-1}}
=\frac{1}{\sqrt{\frac{1}{(1-\varepsilon)^2}r_1^2-1}}\;,$$
as well as
$$\sinh\left[\frac{1+\varepsilon}{2}(c^--u)\right]
=\frac{1}{\sqrt{\coth^2[\frac{1+\varepsilon}{2}(c^--u)]-1}}\geq\frac{1}{\sqrt{\frac{1}{(1-\varepsilon)^2}\left(r(u)\right)^2-1}}\;.$$
Then, for $x$ and $R^{-1}$ sufficiently small, we find that
\begin{equation}
\label{hypToRII}
\frac{\sinh{\left[\frac{1+\varepsilon}{2}(c^{-}-u_1)\right]}}{\sinh{\left[\frac{1+\varepsilon}{2}(c^{-}-u)\right]}}\leq (1+\varepsilon_{x,R})\frac{1+r(u)}{1+r_1}\;.
\end{equation}
Since, according to~\eqref{Gestimate}, the exponential in the left hand side of~\eqref{expG} is bounded by one, using the last estimate together with~\eqref{intExp} and~\eqref{eG0}  we conclude that~\eqref{expG} holds.

Finally, recalling that by definition $r(u)\leq R$, for all $u\leq u_R$,
and since we can choose $2m(1-\varepsilon_{x,R})-p>1$, we get
\begin{eqnarray*}
\int_0^{u_R} \frac{1}{\left(1+r(u)\right)^p}\,e^{m\int_u^{u_1}G(s,r(s))ds}du
&\leq&
\frac{1+\varepsilon_{x,R}}{(1+r_1)^{2m(1-\varepsilon_{x,R})}}\int_0^{u_R} \left(1+r(u)\right)^{2m(1-\varepsilon_{x,R})-p} du
\\
&\leq&
2 \frac{\left(1+R\right)^{2m(1-\varepsilon_{x,R})-p}}{(1+r_1)^{p+1}} u_R
\\
 &\leq&
\frac{C_R u_1}{(1+r_1)^{p+1}} \;,
\end{eqnarray*}
which together with~\eqref{mainEstimate1} concludes the proof of~\eqref{main-estimate}.
\end{proof}
\section{Local existence in both radius and Bondi time}
%
In this section we state a basic local existence result, local both in radius and in Bondi time, that is an essential first step in the our construction of global solutions. Since the proof is fairly  standard we will leave it to Appendix~\ref{apProofLocal}.
\begin{prop}
\label{LocalTimeRadius}
Given $\tau_0>0$ and $h_0\in C^k([0,\tau_0])$, $k\in \mathbb{Z}^+$,
there exists a positive
$$\tau=\tau\left(\sup_{0\leq r\leq \tau_0}|h_0(r)|, \sup_{0\leq r\leq \tau_0}|h'_0(r)|\right)\leq \tau_0$$
and a unique solution, $h\in C^k([0,\tau]^2)$, to
\begin{equation}
\label{mainIVPLocalUR}
 \left\{
\begin{array}{l}
 Dh=G(h-\bar{h})\\
 h(0,r)=h_{0}(r)\;.
\end{array}
\right.
\end{equation}
\end{prop}
%
\section{Global existence in radius and local in Bondi time}
The aim of this section is to upgrade the result in Proposition~\ref{LocalTimeRadius} to a global existence and uniqueness result in radius, but still local in (Bondi) time, at the cost of restricting to small data. More precisely, we will establish the following:
\begin{thm}
\label{localUGlobalR}
Given $0<\delta<1/2$ and $k\in\mathbb{Z}^+$, let $h_0\in C^{k+2}([0,+\infty[)\cap L^{\infty}([0,+\infty[)$ be such that $h'_0\in\Niid_r([0,+\infty[)$ and $h''_0\in\Niiid_r([0,+\infty[)$. Under such conditions, there exists $x_0'=x_0'(\delta)>0$ such that, if
\begin{equation}
\label{smallCond0}
\|h'_0\|_{\Niid_r} \leq x_0' \;,
\end{equation}
then there exists a positive $U=U(\delta,\|h'_0\|_{\Niid_r})$ and
a unique solution, $h\in C^k([0,U]\times[0,\infty[)$,  to
\begin{equation}
\label{mainIVPLocalU}
 \left\{
\begin{array}{l}
 Dh=G(h-\bar{h})\\
 h(0,r)=h_{0}(r)\;.
\end{array}
\right.
\end{equation}
Moreover, $\|\partial_rh\|_{\yupII}$ can be made arbitrarily small by decreasing $x_0'$.
\end{thm}
The proof of this result is based on the coming Lemma \ref{LemmaIteration}, Lemma \ref{lemma-4} and Proposition \ref{propUniq}. The proof relies on the construction of a sequence of functions $\left(h_n\right)_{n\in \mathbb{N}}$, that contracts in an appropriate space. We start by defining
$$
f_n:=f[h_n],\qquad\quad\tilde f_n:=\tilde f[h_n],\qquad\qquad G_n:=G[h_n], \qquad\qquad J_n:=J[h_n],
$$
together with the operators
$$
D_n:=\frac{\partial}{\partial u}-\frac{\tilde f_n}{2}\frac{\partial}{\partial r}
$$
and corresponding characteristics through $(u_1,r_1)$ given by
$$
\chi_n(u)=\chi_n(u;u_1,r_1)=(u; r_n(u;u_1,r_1)).
$$
 The desired sequence is constructed as follows:  Set $w_1(u,r):=h'_0(r)$ 
 and, for $n\geq1$, given~\eqref{D_partial_h} define $w_{n+1}$ as the solution to the linear problem
\begin{equation}
\label{linearDrProblem}
 \left\{
\begin{array}{l}
 D_{n}w_{n+1}=2G_{n}w_{n+1}-J_n\frac{h_n-\bar{h}_{n}}{r}\\
 w_{n+1}(0,r)=h'_{0}(r)\;,
\end{array}
\right.
\end{equation}
with
 \begin{equation}
\label{defhn}
h_{n}(u,r):=h_{U}(u,0)+\int_0^r w_{n}(u,s)ds\;,
 \end{equation}
 where, for a small enough $U>0$, $h_{U}$ is the unique solution in $C^1([0,U]^2)$ to the problem
\begin{equation}
\label{IVPboundary}
 \left\{
\begin{array}{l}
 Dh=G(h-\bar{h})\\
 h(0,r)=h_{0}(r)\;,
\end{array}
\right.
\end{equation}
 as provided by Proposition~\ref{LocalTimeRadius}.

\begin{lem}
\label{LemmaIteration}
Under the conditions of Theorem~\ref{localUGlobalR}, there exists $x_0'=x_0'(\delta)>0$ such that, if
\begin{equation}
\label{bound-h-dot}
\|h'_0\|_{\Niid_r} \leq x_0' \;,
\end{equation}
then there exist constants $C>0$, $x'>0$ and  $x''>0$ for which
\begin{equation}
\|h_n\|_{\yuo}\leq \sup_{0\leq u\leq U} |h_U(u,0)|+ C x'\;,
\end{equation}
\begin{equation}
\label{wBound}
\|w_n\|_{\yupII}\leq  x'\;,
\end{equation}
and
\begin{equation}
\|\partial_rw_n\|_{\yupIII}\leq  x''\,,
\end{equation}
for all $n\in\mathbb{Z}^+$.

Moreover, $x'$ can be made arbitrarily small by decreasing $x_0'$, and $x''$ can be made arbitrarily small by decreasing both $x_0'$ and $\|h''_0\|_{\Niiid_r}$.
\end{lem}

\begin{proof}

The proof is by induction.
For $n=1$ the result follows by noting that
$$\|w_1\|_{\yupII} = \|h'_0\|_{\Niid_r}=:x_0'\;,$$
as well as
$$\|\partial_r w_1\|_{\yupIII} = \|h''_0\|_{\Niiid_r}=:x''_0\;,$$
and that, by setting $b_0:=\sup_{0\leq u\leq U}|h_U(u,0)|$, we have
$$|h_1(u,r)|\leq b_0+ \int_0^r \frac{x'_0}{(1+s)^{2-\delta}} ds = b_0 + C x_0'\;.$$
Now assume, as induction hypothesis, that
\begin{equation}
\|h_n\|_{\yuo}\leq b_0+C x'\;,
\end{equation}
\begin{equation}
\|w_n\|_{\yupII}\leq  x'
\end{equation}
and
\begin{equation}
\|\partial_rw_n\|_{\yupIII}\leq  x''\,,
\end{equation}
for some $C>0$, $x'\geq x'_0$ and $x''\geq x_0''$ to be prescribed during the induction.

Integrating~\eqref{linearDrProblem} along its characteristics leads to
\begin{equation}
\label{drhn}
w_{n+1}(u_{1},r_{1})
 =
 h'_{0}(r_n(0))\,e^{\int_0^{u_1}2 G_n(s,r_n(s))ds}
 -
 \int_0^{u_1}\frac{J_{n}(h_n-\bar{h}_n)(u,r_n(u))}{r_n(u)}e^{\int_u^{u_1}2 G_n(s,r_n(s))ds}du\;,
\end{equation}
and by the triangular inequality
$$\left|w_{n+1}(u_{1},r_{1})\right| \leq
\mathrm{I_0}+\mathrm{I_1}\;,$$
where the definition of $\mathrm{I_0}$ and $\mathrm{I_1}$ should be obvious.

Note that $\partial_r h_n=w_n$. Then, by assuming that $x'$, and consequently $x_0'$, is sufficiently small, we can use~\eqref{expG} to conclude that
\begin{equation}
\label{I0control}
 e^{\int_0^{u_1}2 G_n(s,r_n(s))ds}\lesssim \left(\frac{1+r_n(0)}{1+r_1}\right)^3\;,
\end{equation}
from which we can estimate, using \eqref{bound-h-dot},
 \begin{eqnarray*}
 \label{I0-estimate}
 |(1+r_1)^{2-\delta}\mathrm{I_0}| &\lesssim& (1+r_1)^{2-\delta} \frac{x_0'}{(1+r_n(0))^{2-\delta}}\left(\frac{1+r_n(0)}{1+r_1}\right)^3
 \\
 &\lesssim &
 \left(\frac{1+r_n(0)}{1+r_1}\right)^{1+\delta} x_0'
 \\
 &\lesssim& x_0'\;,
 \end{eqnarray*}
 where to establish the last inequality and \eqref{I0control} we used the fact that, according to~\eqref{Characteristic_ODE} and~\eqref{ftildeLargeR}, all characteristics with a sufficiently large $r_n(0)$ are increasing in $u$.

Also using~\eqref{Jestimate},~\eqref{hhbar}, the induction hypothesis and~\eqref{main-estimate}, with $m=2$ and $p=3-3\delta$, we see that
\begin{eqnarray*}
 |(1+r_1)^{2-\delta}\mathrm{I_1}|&\lesssim& (1+r_1)^{2-\delta} \int_0^{u_1}  \frac{(x')^3}{(1+r_n(u))^{3-3\delta}}e^{\int_u^{u_1}2 G_n(s,r_n(s))ds}du
 \\
 &\lesssim&
 (1+r_1)^{2-\delta} \frac{(x')^3}{(1+r_1)^{4-3\delta}}
  \\
 &\lesssim&
 (x')^3\;.
\end{eqnarray*}
Summing up the last two estimates we conclude that
\begin{equation}
\label{wn-control}
\|w_{n+1}\|_{\yupII}\lesssim x_0' +(x')^3 ,
\end{equation}
which can be made smaller than $x'$ by choosing appropriately small values for $x_0'$ and $x'$.

As an immediate consequence, we can now close the induction for the sequences $(h_n)$ and $(w_n)$ by noting (as before) that
$$|h_{n+1}(u,r)|\leq b_0+ \int_0^r \frac{x'}{(1+s)^{2-\delta}} ds \leq b_0 +  C\,x' \;.$$
To deal with the remaining sequence we derive the following evolution equation directly from \eqref{drbar} and \eqref{linearDrProblem}
\begin{eqnarray}
\nonumber
D_{n}\partial_rw_{n+1}
-3G_{n}\partial_rw_{n+1}
&=&
2\partial_r G_n w_{n+1}-\partial_rJ_n\frac{h_n-\bar{h}_n}{r}
\\
\label{evoldrw}
&\quad&
+\frac{J_n}{r}\left[2\frac{h_n-\bar{h}_n}{r}-w_n\right] \;,
\end{eqnarray}
which when integrated along the corresponding characteristics gives rise to
\begin{equation}
\label{drwn}
 \begin{aligned}
\left|\partial_rw_{n+1}(u_{1},r_{1})\right|
 &\le
 \left|h''_{0}(r_n(0))\right|\,e^{\int_0^{u_1}3 G_n(s,r_n(s))ds}
 \\
 &\quad
 +
 \int_0^{u_1}\left|2\partial_r G_n w_{n+1}\right|(u,r_n(u)) e^{\int_u^{u_1}3 G_n(s,r_n(s))ds}du
 \\
 &\quad
 +
 \int_0^{u_1}\left|\partial_rJ_n\frac{h_n-\bar{h}_n}{r}\right|(u,r_n(u)) e^{\int_u^{u_1}3 G_n(s,r_n(s))ds}du
 \\
 &\quad
 +
 \int_0^{u_1}\left|\frac{J_n}{r}\left[2\frac{h_n-\bar{h}_n}{r}-w_n\right]\right|(u,r_n(u)) e^{\int_u^{u_1}3 G_n(s,r_n(s))ds}du
 \\
&=\mathrm{II_0}+\mathrm{II_1}+\mathrm{II_2}+\mathrm{II_3}\;.
\end{aligned}
\end{equation}
Arguing as in~\eqref{I0-estimate}  we easily obtain
$$|(1+r_1)^{3-\delta}\mathrm{II_0}|\lesssim x_0''\;.$$
To control the remaining terms in \eqref{drwn} we start by noticing that, using~\eqref{defJ},~\eqref{fEstimate},~\eqref{drfEst} and~\eqref{absGestimate}, we get
\begin{equation}
\label{Jr}
 \left|\frac{J_n}{r}\right|\leq 3 \left|\frac{G_n}{r}\right|+3 f +\left|\frac{1+3r^2}{2r}\frac{rx'}{1+r^2}\right|\lesssim 1\;,
\end{equation}
 from which it follows that
\begin{equation}
|\partial_r G_n|= \left|\frac{G_n-J_n}{r}\right|\lesssim 1 \;.
\end{equation}
Then, it becomes clear that from \eqref{wn-control} we get
$$|\partial_r G_n w_{n+1} |(u,r_n(u)) \lesssim \frac{x'}{(1+r_n(u))^{2-\delta}} \;,$$
from which, using~\eqref{main-estimate} with $m=3$ and $p=2-\delta$, we can conclude that
 $$|(1+r_1)^{3-\delta}\mathrm{II_1}| \lesssim (1+r_1)^{3-\delta} \frac{1}{(1+r_1)^{3-\delta}} x' =x'\;.$$
In a similar fashion, by using~\eqref{Jr},~\eqref{hhbar} and the induction hypothesis on $w_n$ we see that
$$\left|\frac{J_n}{r}\left[2\frac{h_n-\bar{h}_n}{r}-w_n\right]\right| (u,r_n(u)) \lesssim \frac{x'}{(1+r_n(u))^{2-\delta}} \;,$$
from which we obtain
$$|(1+r_1)^{3-\delta}\mathrm{II_3}| \lesssim x'\;.$$
To control the final term $\mathrm{II_2}$ we consider the following expression, obtained from~\eqref{defJ} and~\eqref{metricQuoficients},
\begin{equation}
\label{J-control}
J_n=3G_n+3r f_n- \frac{(1-3r^2) f_n(h_n-\bar{h}_n)^2}{4r}\;,
\end{equation}
 and differentiate it with respect to $r$ to get that
 $$|\partial_r J_n|\lesssim 1\;. $$
Then, using~\eqref{hhbar}, we can easily see that
$$\left|\partial_rJ_n\frac{h_n-\bar{h}_n}{r}\right| (u,r_n(u)) \lesssim \frac{x'}{(1+r_n(u))^{2-\delta}}\;,$$
which gives
$$|(1+r_1)^{3-\delta}\mathrm{II_2}| \lesssim x'\;.$$
Summing up the relevant estimates we find
\begin{equation}
\label{drwControl}
 \|\partial_r w_{n+1}\|_{\yupIII} \lesssim x_0'' + x' =: x''\;,
\end{equation}
which concludes the proof of the lemma.
 \end{proof}
We are now ready to establish the following:
\begin{lem}
\label{lemma-4}
For $\|h_0'\|_{\yupII}$ sufficiently small, the sequence $(h_n)$ contracts in $\yuo$.
\end{lem}

\begin{proof}
A straightforward computation reveals that the difference between successive terms of the sequence $(w_n)$ satisfies the evolution equation
\begin{eqnarray}
\nonumber
D_{n}(w_{n+1}-w_n)
-2G_{n}(w_{n+1}-w_n)
&=&
2 (G_n-G_{n-1}) w_{n}+\frac{1}{2}(\tilde f_n-\tilde f_{n-1})\partial_rw_n
\\
\label{evoldrw}
&\quad&
-\frac{J_n}{r}\left( (h_n-\bar{h}_n)- (h_{n-1}-\bar{h}_{n-1})\right)\\
&\quad&
-\frac{J_n-J_{n-1}}{r}(h_{n-1}-\bar h_{n-1})\nonumber
\;.
\end{eqnarray}
After noticing that $(w_{n+1}-w_n)(0,r)\equiv 0$, the integration of the previous equation along the corresponding characteristics gives rise to
\begin{equation}
\label{difwn}
 \begin{aligned}
\left|w_{n+1}-w_n\right|(u_{1},r_{1})
 &\le
\int_0^{u_1}\left|2 (G_n-G_{n-1}) w_{n}\right|(u,r_n(u)) e^{\int_u^{u_1}2 G_n(s,r_n(s))ds}du
 \\
 &\quad
 +
 \int_0^{u_1}\left|\frac{1}{2}(\tilde f_n-\tilde f_{n-1})\partial_rw_n \right|(u,r_n(u)) e^{\int_u^{u_1}2 G_n(s,r_n(s))ds}du
 \\
 &\quad
 +
 \int_0^{u_1}\left|\frac{J_n}{r}\left( (h_n-\bar{h}_n)- (h_{n-1}-\bar{h}_{n-1})\right)\right|(u,r_n(u)) e^{\int_u^{u_1}2 G_n(s,r_n(s))ds}du
 \\
  &\quad
 +
 \int_0^{u_1}\left| \frac{J_n-J_{n-1}}{r}(h_{n-1}-\bar h_{n-1})\right| (u,r_n(u)) e^{\int_u^{u_1}2 G_n(s,r_n(s))ds} du.
\end{aligned}
\end{equation}
Now, from~\eqref{hhbar} and \eqref{wBound} we have
\begin{equation*}
\left|(h_{n}-\bar{h}_{n})+(h_{n-1}-\bar{h}_{n-1})\right|
\lesssim
\frac{r}{(1+r)^{2-\delta}}\,x' \;,
\end{equation*}
while  using~\eqref{hhbar2} we get
\begin{equation*}
\left|(h_{n}-\bar{h}_{n})-(h_{n-1}-\bar{h}_{n-1})\right|
\leq 2 \|h_n-h_{n-1}\|_{\yuo}\;,
\end{equation*}
so that
\newcommand{\hMh}{\left\|h_{n}-h_{n-1}\right\|_{\yuo}}
\begin{eqnarray}
\nonumber
\left|(h_{n}-\bar{h}_{n})^2-(h_{n-1}-\bar{h}_{n-1})^2\right|
\label{hhbarSq}
&\lesssim& \frac{r}{(1+r)^{2-\delta}}\,x'\, \left\|h_{n}-h_{n-1}\right\|_{\yuo}\;,
\end{eqnarray}
and then, from~\eqref{metricQuoficients} and~\eqref{hhbar}, we find
\begin{eqnarray}
\nonumber
|f_{n}-f_{n-1}|
&=&
\left|\exp\left(\frac{1}{2}\int_0^r\frac{(h_{n}-\bar{h}_{n})^2}{s}\right)
-\exp\left(\frac{1}{2}\int_0^r\frac{(h_{n-1}-\bar{h}_{n-1})^2}{s}\right)\right|
\\
\nonumber
&\lesssim &
\frac{1}{2}\int_0^r\frac{\left|(h_{n}-\bar{h}_{n})^2-(h_{n-1}-\bar{h}_{n-1})^2\right|}{s}ds
\\
\nonumber
&\lesssim&
\,x'\,\hMh\int_0^r\frac{1}{(1+s)^{2-\delta}}ds
\\
\label{deltaf}
&\lesssim&
\,x'\,\hMh \;.
\end{eqnarray}
Therefore, by recalling~\eqref{metricQuoficients}, we have
\begin{equation}
\label{deltaTildef}
\begin{aligned}
|\tilde{f}_{n}-\tilde{f}_{n-1}|
&=
\left|\frac{1}{r}\int^{r}_{0}(f_{n}-f_{n-1})(1-3 s^2)ds\right|
\\
&\lesssim
\,x'\,\hMh (1+r)^2\;,
\end{aligned}
\end{equation}
while from~\eqref{G} and \eqref{metricQuoficients}, we get
\begin{equation}
\begin{aligned}
\label{deltaG}
|G_{n}-G_{n-1}|
&=
\frac{1}{2r}\left|(f_{n}-f_{n-1})(1- 3 r^2)-\frac{1}{r}\int^{r}_{0}(f_{n}-f_{n-1})(1-3 s^2)ds\right|
\\
&\lesssim
\,x'\,\hMh (1+r)\;.
\end{aligned}
\end{equation}
In turn, from \eqref{J-control}, \eqref{hhbar} and the previous estimates, we obtain
\begin{equation}
\begin{aligned}
\label{deltaG}
\left| J_{n}-J_{n-1}\right|
&=
\left| 3 (G_n-G_{n-1}) +3r(f_n-f_{n-1})\right.\\
&\quad\left.-\frac{1-3r^2}{4r}\left[(f_n-f_{n-1})(h_n-\bar h_n)^2+f_{n-1}\left((h_n-\bar h_n)^2-(h_{n-1}-\bar h_{n-1})^2\right)\right]\right|
\\
&\lesssim
\,x'\,\hMh (1+r)\;.
\end{aligned}
\end{equation}
If we recall~\eqref{main-estimate},~\eqref{Jr} and argue as in the proof of Lemma~\ref{LemmaIteration}, it becomes clear from \eqref{difwn} that
\begin{equation}
\label{alpha}
\left|w_{n+1}-w_n\right|(u_{1},r_{1})
\lesssim \,(1+x'')x'\,\hMh \frac{1}{(1+r)^{2-\delta}}\;,
\end{equation}
from which we immediately see that
\begin{eqnarray*}
\left|h_{n+1}-h_n\right|(u,r)
&\leq & \int_0^r  \left |w_{n+1}-w_n\right| ds
\\
&\lesssim& \,(1+x'')x'\,\hMh \int_0^r   \frac{1}{(1+s)^{2-\delta}}ds
\\
&\lesssim& \,(1+x'')x'\,\hMh\;.
\end{eqnarray*}
The desired result then follows by decreasing $x'$, which according to Lemma~\ref{LemmaIteration} can be achieved by decreasing $\|h_0'\|_{\yupII}$ and recalling by that $x''$ is a quantity that decreases with $x'$~(see \eqref{drwControl}).
\end{proof}
We have thus concluded that the sequence $(h_n)$ converges uniformly to some continuous function $h:[0,U]\times[0,\infty[\rightarrow \mathbb{R}$. It follows that $f_n=f[h_n]\rightarrow f:=f[h]$, in $\yuo$, by obtaining an estimate similar to~\eqref{deltaf} with $f_{n-1}$ replaced by $f$ and $h_{n-1}$ replaced by $h$. In a similar manner, by adapting~\eqref{deltaTildef}, ~\eqref{deltaG} and~\eqref{Jr} we can conclude that $\tilde f_n\rightarrow \tilde f$, $G_n\rightarrow G$  and $J_n\rightarrow J$, uniformly in  domains of the form $[0,U]\times[0,R]$. Then we can also conclude that each characteristic $r_n(\,\cdot\,; u_1,r_1)$ also converges uniformly, in $[0,U]\times[0,R]$, to a characteristic $r(\,\cdot\,; u_1,r_1)$ of $D=\partial_u-\frac{1}{2}\tilde f[h]\partial_r$. Considering this last fact, we can recycle the proof after~\eqref{deltaChi}. Then, by using~\eqref{hhbar} and~\eqref{drhn}, we see that the sequence $(w_n)$ converges uniformly, in every domain of the form $[0,U]\times[0,R]$, to the function
\begin{equation}
\label{wSol}
w(u_{1},r_{1})
 =
 h'_{0}(r(0))\,e^{\int_0^{u_1}2 G(s,r(s))ds}
 -
 \int_0^{u_1}\frac{J(h-\bar{h})(u,r(u))}{r(u)}e^{\int_u^{u_1}2 G(s,r(s))ds}du\;.
\end{equation}
 The previous function is clearly a continuous solution of
$$Dw=2Gw-J \frac{h-\bar{h}}{r}\;. $$
From~\eqref{defhn} we can now also conclude that
\begin{equation}
\label{hInt}
 h(u,r)=h_U(u,0)+\int_0^r w(u,s)\,ds\;,
\end{equation}
which satisfies
$$\partial_r h =w\;. $$
Then, we are allowed to differentiate~\eqref{wSol} with respect to $r_1$ and conclude that $w\in C^1$. Consequently we can also differentiate~\eqref{hInt} with respect to $u$ to conclude that $h$ is also $C^1$. In this process we established that $\partial_rw=\partial_r^2h$ is continuous.

By differentiating~\eqref{hInt} along $D$ we get
$$Dh(u,r) = \partial_u h_U(u,0) -\frac{1}{2} \tilde f(u,r) w(u,r) + \int_0^r\partial_u w(u,s) ds $$
with
\begin{eqnarray}
 \int_0^r\partial_u w(u,s) ds &=& \int_0^r Dw(u,s)+ \frac{1}{2} \tilde f \partial_r w (u,s) ds
\\
&=&
\int_0^r 2G\partial_r h(u,s)-\frac{J(h-\bar{h})(u,s)}{s} + \frac{1}{2} \tilde f \partial^2_r h (u,s) ds
\\
&=&
\int_0^r \partial_r \left(  \frac{1}{2} \tilde f \partial_r h  + G(h-\bar{h}) \right) (u,s) ds\;.
\end{eqnarray}
Thus, using the fact that $(h-\bar{h})(u,0)=0$, which follows from~\eqref{hhbar}, and recalling that $\tilde f(u,0)=1$, we find
\begin{equation}
 \label{eqh0}
 Dh(u,r) = G(h-\bar{h}) (u,r) + \frac{1}{2} \partial_r \left( h_U(u,0)-h(u,0)\right)\;,
\end{equation}
where we took into account that  $D [h_U](u,0)=G [h_U](u,0)(h_U-\bar h _U)(u,0)=0$ to conclude that $\partial_u h_U(u,0)=\frac{1}{2}\partial_r h_U(u,0)$.

To finish the proof of Theorem~\ref{localUGlobalR} we need the following uniqueness result:
\begin{prop}
\label{propUniq}
 Given $h_U\in C^1([0,1]^2)$ and $h_0\in C^2([0,1])$, with sufficiently small  $\|h_0'\|_{L^{\infty}_r}$, then there exists $0<U\leq 1$, such that the initial value problem
\begin{equation}
\label{uniqIVP}
 \left\{
\begin{array}{l}
Dh(u,r) = G(h-\bar{h}) (u,r) + \frac{1}{2}\partial_r  \left( h_U(u,0)-h(u,0)\right) \\
 D\partial_rh=2G\partial_r h-J\frac{h-\bar{h}}{r}\\
 h(0,r)=h_{0}(r) \\
 \partial_rh(0,r)=h'_{0}(r)\;,
\end{array}
\right.
\end{equation}
 admits at most one solution in $C^1([0,U]^2)$  with continuous second radial derivative.
\end{prop}

\begin{proof}
 Let $h_1$ and $h_2$ be solutions to the initial value problem under analysis, that satisfy the prescribed regularity conditions. Then
\begin{eqnarray}
\nonumber
D_{1}(h_{2}-h_1)
-G_{1} (h_{2}-h_1)
&=&
(G_2-G_{1}) (h_2-\bar{h}_2)-G_1(\bar{h}_2-\bar{h}_1)
\\
\label{evolUniq}
&\quad&
+\frac{1}{2}\left( \tilde f_2-\tilde f_1\right)\partial_rh_2 -\frac{1}{2}\left(\partial_rh_2(u,0)-\partial_rh_1(u,0)\right)
\;.
\end{eqnarray}
Arguing as in the proof of the previous lemma, we can control all differences of the form $|A_2-A_1|$ in the previous evolution equation, by estimates of the form
$$|A_2-A_1|\lesssim \|h_2-h_1\|_{\yuo}\;.$$
Note that to control the difference $\partial_rh_2(u,0)-\partial_rh_1(u,0)$ we need to obtain an evolution equation for $\partial_rh_2(u,r)-\partial_rh_1(u,r)$ analogous to~\eqref{difwn}, with $w=\partial_r h$ and $n=1$,  which will lead to an estimate of the form
$$|\partial_r h_2-\partial_r h_1|\lesssim \|h_2-h_1\|_{\yuo}\;.$$
The previous method requires to have a bound on the second radial derivatives of $h_1$, whose existence follows from our regularity assumptions.

Integrating~\eqref{evolUniq} along the corresponding characteristics then leads to
$$|h_2-h_1|(u_1,r_1) \leq C \int_0^{u_1} \|h_2-h_1\|_{\yuo} ds = Cu_1  \|h_2-h_1\|_{\yuo}\;,$$
which, for $u_1\leq U$ with small enough $U$,  implies
$$(1-C\,U) \|h_2-h_1\|_{\yuo} \leq 0 \Rightarrow \|h_2-h_1\|_{\yuo}=0\;.$$
\end{proof}
It now follows that $h=h_U$, in $[0,U]^2$, for a small enough $U>0$ and, consequently, since $h$ satisfies~\eqref{eqh0}, then it in fact satisfies~\eqref{mainIVPLocalU}.

To conclude the proof of Theorem~\ref{localUGlobalR} we just need to notice that the higher regularity claims follow from differentiating the integral version of~\eqref{mainEq}  (compare with~\cite[(17)]{CostaSpherically}) and the proof of the uniqueness claim is similar to the proof of Proposition~\ref{propUniq} (compare with~\cite[(59)]{CostaProblem}).

\section{Global well posedness and decay of solutions}
%
\begin{thm}
\label{thm}
\label{global-thr}
Given $0<\delta<1/2$ and $k\in\mathbb{Z}^+$, let $h_0\in C^{k+2}([0,+\infty[)\cap L^{\infty}([0,+\infty[)$ be such that $h'_0\in\Niid([0,+\infty[)$ and $h''_0\in\Niiid([0,+\infty[)$. Under such conditions, there exists $\tilde x_0=\tilde x_0(\delta)>0$ such that, if
\begin{equation}
\label{smallCond0-1}
\|h'_0\|_{\Niid_r} \leq \tilde x_0 \;,
\end{equation}
then the initial value problem
\begin{equation}
\label{mainEq0II}
 \left\{
\begin{array}{l}
  Dh = G\left(h-\bar{h}\right) \\
  h(0,r)= h_0(r)\;,
\end{array}
\right.
\end{equation}
has a unique (global) solution $h\in C^{k}([0,+\infty[\times[0,+\infty[)$.

Moreover
\begin{equation}
\label{boundh2II}
\|\partial_r h(u,\,\cdot\,)\|_{\Niid_r}\leq Ce^{-(1+\delta/2) u}\;,
\end{equation}
with $C>0$ and, given $R>0$, if $\tilde x_0 \leq\underline{x}(R)$, with the later sufficiently small, then there exists $C_R>0$ such that
\begin{equation}
\label{localBound}
\sup_{r\leq R} |\partial_r h(u,\,\cdot\,)|\leq C_R\, e^{-2\,u}\;.
\end{equation}
\end{thm}

\begin{Remark}
For general $\Lambda>0$ the exponents in~\eqref{boundh2II} and~\eqref{localBound} should be multiplied by $H:=\sqrt{\frac{\Lambda}{3}}$.
\end{Remark}

\begin{proof}

\newcommand{\cE}{\mathcal E}
\newcommand{\cF}{\mathcal F}

Let us start by considering $\tilde x_0<x_0'$, with $x_0'$ taken from~\eqref{smallCond0}. Then  Theorem~\ref{localUGlobalR} guarantees the existence of $\tilde U=U(\tilde x_0)>0$ and $h\in C^1([0,\tilde U]\times[0,+\infty[)$ solving~\eqref{mainEq0II}.
Now let $U^*\in[\tilde U,\infty[$ be the corresponding maximal time of existence and define, for $0\leq u< U^*$ and $R\gg1$,
\begin{equation}\label{energy1}
\cE_R(u)=\sup_{r\in[0,R]}|(1+r)^{2-\delta}\partial_r h(u,\,\cdot\,)|\;.
\end{equation}
Note that, since the supremum is over a compact set, $\cE_R$ is continuous in $[0,U^*[$.   

For $x'\in]\tilde x_0, x_0'[$, to be specified independently of $R$ during the proof, define also
\begin{equation}\label{cU}
{\cal U}_R:=\left\{u_1\in [0,U^*[ \,:\, \sup_{u\in[0,u_1]}\cE_R(u)\leq x'\right\},
\end{equation}
which is non-empty by Theorem~\ref{localUGlobalR} and closed in view of the continuity of $\cE_R$.

To show that the previous set is also open we integrate~\eqref{D_partial_h} to get
\begin{equation}
\label{drh}
\partial_r h(u_{1},r_{1})=
\partial_r h_0(\chi(0))\,e^{\int_0^{u_{1}}2 G_{|\chi}dv}
-\int_0^{u_1}\left(J\frac{h-\bar{h}}{r}\right)_{|\chi}e^{\int^{u_1}_{u}2 G_{|{\chi}}dv}du\;,
\end{equation}
where $\chi(u)=\chi(u;u_1,r_1)=(u,r(u;u_1,r_1))$ is the characteristic through $(u_1,r_1)$.
Now, according to~\cite[(30),(31)]{CostaProblem}, there exists $r^{-}_{c}=1-\varepsilon$ (where once again $\varepsilon$ is a quantity that can be made arbitrarily small by decreasing $x'$) such that, if $r_1>  r_c^-$ then 
$$r(u;u_1,r_1)>r_c^-,$$
for all $u\leq u_1$, and if $r_1\leq  r_c^-$ then
\begin{equation}
\label{charExplicit2}
r(u;u_1,r_1)\geq (1-\varepsilon)\tanh \left(\frac{1+\varepsilon}{2}(c^--u)\right)\;,
\end{equation}
with $c^-$ an integration constant. It follows that
\begin{equation*}
 -\int^{u_{1}}_{u}r(v)dv
 \leq
\frac{2(1-\varepsilon)}{1+\varepsilon}\log{\left(\frac{\cosh{\left(\frac{1+\varepsilon}{2}(c^--u_1)\right)}}{\cosh{\left(\frac{1+\varepsilon}{2}(c^--u)\right)}}\right)}
\leq 2(1-\varepsilon) \log{\left(2e^{\frac{1+\varepsilon}{2} (u-u_1)}\right)}
\end{equation*}
and then~\eqref{Gestimate} gives
\begin{equation}\label{intN}
 e^{\int^{u_1}_{u}2 G_{|{\chi}}dv}\lesssim  e^{2(1-\varepsilon)(u-u_1)}\;.
\end{equation}
For $r_1<r_c^-$ and $u\leq u_1<U^*$, it follows from~\cite[(30)]{CostaProblem} that $r(u)\lesssim 1$, so that using \eqref{hhbar},~\eqref{Jestimate},~\eqref{energy1},~\eqref{cU} and \eqref{intN} applied to~\eqref{drh} gives
\begin{equation}
\label{energyEst1}
|(1+r_1)^{2-\delta}\partial_r h(u_{1},r_{1})|\lesssim
\cE_R(0)\,e^{-2(1-\varepsilon)u_1}
+ x' \int_0^{u_1}\cE_R(u)e^{2(1-\varepsilon)(u-u_1)}du\;.
\end{equation}
We now consider $r_1>r_c^-$ and start by noticing that, since in such case $r(u)\geq r_c^-=1-\varepsilon$, for all $u\leq u_1$, then inequality~\eqref{Gestimate} gives
\begin{equation}
\label{Gestimate2}
G(u,r(u))\leq -(1-\varepsilon)\;.
\end{equation}
It is then easy to see that the condition $r(u)\leq R$ immediately leads to~\eqref{energyEst1} with the implicit constant depending on $R$. This is relevant for the localized estimate \eqref{localBound}. 

However we also need an estimate which is uniform in $R$, which will arise at the cost of weakening the decay in $u$.  
With that goal in mind we write $G=qG+(1-q)G$, with $q\in[0,1]$, and then by using~\eqref{Gestimate2},~\eqref{expG},~\eqref{intExp} and~\eqref{hypToRII} we arrive at
\begin{eqnarray}
\nonumber
e^{2\int_u^{u_1}G(s,r(s))ds}
&=&
e^{2q\int_u^{u_1}G(s,r(s))ds}e^{2(1-q)\int_u^{u_1}G(s,r(s))ds}
\\
\label{intGII}
&\leq&
C\left(\frac{1+r(u)}{1+r_1}\right)^{4q-\varepsilon}e^{2(1-q) (1-\varepsilon)(u-u_1)}\;.
\end{eqnarray}
In such case, recallling \eqref{Jestimate} and \eqref{hhbar}, we get
\begin{eqnarray*}
|(1+r_1)^{2-\delta}\partial_r h(u_{1},r_{1})|
&\lesssim&
 \left(\frac{1+r_1}{1+r(0)}\right)^{2-\delta-4q+\varepsilon}\cE_R(0)e^{-2(1-q)(1-\varepsilon) u_1}
\\
&\quad&
+\,x'\int_0^{u_1}  \frac{\left(1+r_1\right)^{2-\delta-4q+\varepsilon}}{\left(1+r(u)\right)^{3-3\delta-4q+\varepsilon}}\, \cE_R(u)e^{2(1-q)(1-\varepsilon)(u-u_1)}du\;.
  \end{eqnarray*}
By choosing $q$ such that
$$2-\delta-4q+\varepsilon\leq 0\Leftrightarrow q\geq \frac{2-\delta+\varepsilon}{4}\;,$$
and recalling $0<\delta<1/2$, we conclude that, for $r_1\leq R$,
\begin{equation}
\label{energyEst2}
|(1+r_1)^{2-\delta}\partial_r h(u_{1},r_{1})|\lesssim
\cE_R(0)\,e^{-2(1-q)(1-\varepsilon)u_1}
+ x' \int_0^{u_1}\cE_R(u)e^{2(1-q)(1-\varepsilon)(u-u_1)}du\;.
\end{equation}
So by setting
\newcommand{\dec}{\hat H}
$$\dec:=2(1-q)(1-\varepsilon)\;,$$
since $1-q\leq 1$, we conclude that the estimates~\eqref{energyEst1} and~\eqref{energyEst2} give rise to the following estimate, valid for all $r_1\leq R$,
\begin{equation}
\label{energyEst3}
|(1+r_1)^{2-\delta}\partial_r h(u_{1},r_{1})|\lesssim
\cE_R(0)\,e^{-\dec u_1}
+ x' \int_0^{u_1}\cE_R(u)e^{\dec(u-u_1)}du\;.
\end{equation}
By taking the supremum in $r_1$ we then obtain
\begin{equation}
\label{energyEst4}
\cE_R(u_1)\lesssim
\cE_R(0)\,e^{-\dec u_1}
+ x' \int_0^{u_1}\cE_R(u)e^{\dec(u-u_1)}du\;,
\end{equation}
and if we now define $\cF(u):=e^{\dec u}\cE_R(u)$, the previous estimate translates into
\begin{equation}
\label{energyEst4}
\cF(u_1)\lesssim
\cE_R(0)
+ x' \int_0^{u_1}\cF(u)du\;.
\end{equation}
Applying Gronwall's inequality this leads to
$$\cF(u_1) \leq C_1 \cE_R(0) e^{C_2 x' u_1}$$
which, in terms of the quantity $\cE_R$, is
\begin{equation}
\label{energyControl}
\cE_R(u_1) \leq C_1 \cE_R(0) e^{-(\dec - C_2 x') u_1}\leq  C_1 \tilde x_0 e^{-(\dec - C_2 x') u_1} \;,
\end{equation}
where $C_1, C_2>0$ are constants. We can now choose $x'$ small enough in order to guarantee that $\dec - C_2 x'>0$ and then choose $\tilde x_0$ satisfying $C_1\tilde x_0\leq \frac{1}{2}x'$. For such choices, it becomes clear that ${\cal U}_R$ is open and since it is also closed and non-empty we conclude that ${\cal U}_R=[0,U^*[$. This holds for an arbitrary  $R>>1$, so we find that, for all $r_1>>1$, 
$$|(1+r_1)^{2-\delta}\partial_r h(u_{1},r_{1})|\leq \cE_{r_1}(u_1) \leq x'$$
and since $x'$ does not depend on the choice of $r_1$, by taking suprema of both sides of the last inequality we arrive at
\begin{equation}
\label{globalE}
\cE(u):= \|\partial_r h(u,\,\cdot\,)\|_{\Niid_r}<x' <x_0'\;\;,\; \text{ for all }u\in[0,U^*[\;.
\end{equation}
If we assume that $U^*<\infty$ and choose $0\leq \varepsilon<U'=U(x')$, we can use Theorem~\ref{localUGlobalR} to solve~\eqref{mainEq0II}, in $[U^*-\varepsilon, U^*-\varepsilon+U']\times [0,\infty[$, with initial data provided by $h(U^*-\varepsilon, \,\cdot\,)$. Concatenating the two solutions we obtain a solution with existence time  $U^*-\varepsilon+U'>U^*$, in contradiction with the definition of $U^*$. In conclusion $U^*=\infty$.

The regularity and uniqueness claims follow as in the proof of Theorem~\ref{localUGlobalR}.
In turn, the decaying estimates~\eqref{boundh2II} and~\eqref{localBound}, with an $\varepsilon$ loss, follow immediately from the derived estimates for $\cE_R$, which, in view of~\eqref{globalE}, also hold if we replace $\cE_R$ by $\cE$.
To remove the $\varepsilon$ loss we just need to reuse the final  argument in the proof of~\cite[Theorem 4]{CostaProblem}.

\end{proof}


\section{Proof of Theorem \ref{main-thm}}
\label{final-sec}

\begin{proof} [Proof of Theorem \ref{main-thm}] 
	
Existence and uniqueness are an immediate consequence of Theorem~\ref{thm} (recall also Remark~\ref{remark-final-sol}). The statement about geodesic completeness follows from the estimates~\eqref{phi-convergence}--\eqref{metric-convergence3} which, as we will now show, are a consequence of~\eqref{boundh2II} and~\eqref{localBound}:

From \eqref{hhbar} and \eqref{absGestimate} we get
\begin{equation}
 |Dh|=|G(h-\bar h)| \lesssim r^{\delta}\;,
\end{equation}
and then, in view of~\eqref{localBound}, for $r\leq R$ we have
\begin{equation}
 |\partial_uh|=|Dh+\frac{1}{2}\tilde f \partial_r h| \lesssim C_R e^{-2 u}\;. 
\end{equation}
Since the last estimate is integrable in $u$, there exists $\underline{h}(\infty,r)\in\mathbb{R}$ such that 
$$\lim_{u\rightarrow \infty} h(u,r)=\underline{h}(\infty,r)\;.$$
In fact, for $r_1\leq r_2\leq R$,
$$|\underline{h}(\infty,r_2)-\underline{h}(\infty,r_1)|\leq \lim_{u\rightarrow \infty} \int_{r_1}^{r_2}|\partial_rh(u,\rho)|d\rho \leq \lim_{u\rightarrow \infty} C_R e^{-2u}=0\;,$$
and therefore $\underline{h}(\infty,r)\equiv\underline{h}(\infty)$.  Moreover
$$|h(u,r)-\underline{h}(\infty)|=\left|\int_u^{\infty}\partial_u h(s,r) ds\right|\leq C_R e^{-2u}\;$$
and 
$$|\phi(u,r)-\underline{h}(\infty)|\leq |h(u,r)-\bar{h}(u,r)|+|h(u,r)-\underline{h}(\infty)|\leq C_R e^{-2u}\;,$$
so we set $\underline{\phi}(\infty)=\underline{h}(\infty)$ and~\eqref{phi-convergenceR} follows. 

For the convergence of the metric components~\eqref{metric-convergenceR} and~\eqref{metric-convergenceR2} we just need to notice that, for $r\leq R$, in view of~\eqref{hhbar} and~\eqref{localBound}, we have 
\begin{eqnarray*}
 |f(u,r)-1| &=& |e^{\frac{1}{2}\int_0^r\frac{(h-\bar h)^2}{s}ds}-1| \\
 &\lesssim&  \sup\{f(u,r)+1\} \int_0^r\frac{(h-\bar h)^2}{s}ds \\
 &\leq& C_R  \int_0^r\frac{1}{(1+s)^{3-2\delta}}e^{-4 u}ds \\
 &\leq&  C_R e^{-4 u}\;,
\end{eqnarray*}
and then
\begin{eqnarray*}
 |\tilde f(u,r)-(1-r^2)| &\leq &\frac{1}{r} \int_0^r|(1-3s^2)||f(u,s)-1| ds \\
 &\leq& C_R e^{-4u}  \;.
\end{eqnarray*}
 Now, since $\partial_rh$ is integrable in the entire radial range, then there exists $\underline \phi(u)$ such that 
 $$\lim_{r\rightarrow\infty} h(u,r)=\underline \phi(u)\;.$$
 Using~\eqref{boundh2II}, we immediately see that
$$|h(u,r)-\underline{\phi}(u)|\leq \int_r^{\infty}|\partial_rh(u,s)|ds \lesssim \frac{1}{(1+r)^{1-\delta} }e^{-(1+\delta/2)u}$$ 
and~\eqref{phi-convergence} follows. Note that we can use the previous estimate to conclude that $\underline \phi$ is continuous.

Since $f$ is bounded and monotone in the radial variable, we can set  $f(u,\infty)=\lim_{r\rightarrow\infty}f(u,r)\in\mathbb{R}$. Then we define a new Bondi time coordinate by setting 
\begin{equation}
d\hat u = f(u,\infty) du\;,
\end{equation}
which in view of~\eqref{fEstimate} can be chosen to satisfy
\begin{equation}
\label{uhat}
u\leq \hat u \leq (1+\varepsilon) u\;,
\end{equation}
with $\varepsilon>0$ a constant that can be made arbitrarily small by decreasing $\|h'_0\|_{L^{\infty,2-\delta}_r}$.

The spacetime metric then becomes
\begin{equation}
\label{metricBondiHat}
 \metric=-\hat f(\hat u,r)\tilde{\hat f}(\hat u,r)d\hat u^{2}-2\hat f(\hat u,r)d\hat udr+r^{2}\sigma_{\mathbb{S}^2}\;,
\end{equation}
with 
$$\hat f(u,r)=\frac{ f(u,r)}{f(u,\infty)}$$
and
$$\tilde{\hat f} (u,r)=\frac{ \tilde f(u,r)}{f(u,\infty)}\;.$$
Note that the new coordinate was designed so that $\hat f(\hat u,\infty)\equiv 1$. 

We then have
\begin{eqnarray}
 \nonumber
 |\hat f(u,r)-1| &=& \frac{1}{f(u,\infty)}|f(u,r)-f(u,\infty)| \\
 \nonumber
 &\lesssim&  \left|e^{\frac{1}{2}\int_0^r\frac{(h-\bar h)^2}{s}ds}-e^{\frac{1}{2}\int_0^{\infty}\frac{(h-\bar h)^2}{s}ds}\right|  \\
  \nonumber
 &\lesssim&  \int_r^{\infty}\frac{1}{(1+s)^{3-2\delta}}e^{-2(1+\delta/2) u}ds \\
  \nonumber
 &\lesssim&  \frac{1}{(1+r)^{2(1-\delta)}}e^{-2(1+\delta/2) u} \\
 \label{metricHat}
 &\lesssim&  \frac{1}{(1+r)^{2(1-\delta)}}e^{-2(1-\varepsilon)(1+\delta/2) \hat u}\;,
\end{eqnarray}
and
\begin{eqnarray}
 \nonumber
 |\tilde{\hat f}(u,r)-(1-r^2)| &\leq &  \frac{1}{f(u,\infty)}|\tilde f(u,r)-(1-r^2)f(u,\infty)| \\
  \nonumber
 &\lesssim& \frac{1}{r} \int_0^r|(1-3s^2)||f(u,s)-f(u,\infty)| ds \\
  \label{metricHat2}
 &\lesssim& (1+r)^{2\delta}e^{-2(1-\varepsilon)(1+\delta/2) \hat u}  \;.
\end{eqnarray}
Now we can construct a diffeomorphism between our (dynamic) spacetime and the de Sitter spacetime by identifying the points with the same $(\hat u,r,\omega)$ coordinates. In that case, we write the de Sitter metric in the form   
$$ \metric^{\mathrm {dS}}=(r^2-1)\left(d \hat u-\frac{1}{r^2-1}dr \right)^2-\frac{1}{r^2-1}dr^2 + r^2 \sigma_{\mathbb{S}^2}\;,$$
set
$$e_0 =\frac{1}{\sqrt{r^2-1}}\partial_{\hat u}+\sqrt{r^2-1}\,\partial_{\hat r}\;,\;\;\;\;\;\;\;\;e_1 =\frac{1}{\sqrt{r^2-1}}\partial_{\hat u}$$
 and fix an orthonormal frame $(e_A)_{A=2,3}$ of $(\mathbb{S}^2, r^2 \sigma_{\mathbb{S}^2})$. Then $(e_0,e_1,e_A)$ forms an orthonormal frame of de Sitter and~\eqref{metric-convergence3} follows from~\eqref{metricHat} and~\eqref{metricHat2}.
\end{proof}
\appendix
\section{Proof of Proposition~\ref{LocalTimeRadius}}
\label{apProofLocal}

\begin{proof}
The proof relies on the construction of a $\|\cdot\|_{\cal U}$-contracting sequence $(h_n)$ where
$$\|w\|_{\cal U}=\sup_{(u,r)\in {\cal U}} |w(u,r)|\;,$$
and, for a given $m>0$ and $0\leq \tau\leq \tau_0$, to be specified during the proof, the domain ${\cal U}$ is of the form
\begin{equation}
\label{defU}
{\cal U}=\{(u,r)\,:\, 0\leq u\leq \tau \;,\; r\leq m(\tau-u)+\tau \}  \;.
\end{equation}
We define $h_1:{\cal U}\rightarrow \mathbb{R}$ by
$$h_1(u_1,r_1)=h_0(r_1)\;,$$
and note that $\|h_1\|_{{\cal U}}\leq \sup_{0\leq r\leq \tau_0}|h_0(r)|$ and $\|\partial_rh_1\|_{{\cal U}}\leq \sup_{0\leq r\leq \tau_0}|h'_0(r)|$.
Then the desired sequence can be constructed by setting
\begin{equation}
h_{n+1}(u_1,r_1)=h_{0}(r_n(0))e^{\int_{0}^{u_1}G_{n|_{\chi_{n}}}dv}
-\int_{0}^{u_1}\left(G_{n}\bar{h}_{n}\right)_{|_{\chi_{n}}}e^{\int^{u_1}_{u}G_{n|_{\chi_{n}}}dv}du\;,
\label{h_{n+1}_integral}
\end{equation}
where the characteristic
\begin{equation}
\label{char}
\chi_n(u;u_1,r_1)=(u,r_n(u;u_1,r_1))\;,
\end{equation}
is the integral curve of $D_n:=\partial_u-\frac{1}{2}\tilde f_n\partial_r$
through the point $(u_1,r_1)$, i.e.,
$r_n$ is the unique solution to
\begin{equation}
\label{charEq}
\frac{dr_n}{du}=-\frac{1}{2}\tilde f_n(u,r_n(u))\;\;,\; r_n(u_1)=r_1\;,
\end{equation}
with $\tilde f_n=\tilde f[h_n]$ and $G_n=G[h_n]$ defined by using \eqref{metricQuoficients} and~\eqref{G}.

Now, assume as induction hypothesis that, there exists $m,\tau>0$ such that $h_n\in C^1({\cal U})$ with
\begin{equation}
\label{induction}
 \|h_n\|_{{\cal U}}\leq C_d\;\;,\;\;\|\partial_rh_n\|_ {{\cal U}}\leq  C_d\;,
\end{equation}
for some $ C_d=C_d(\sup_{0\leq r\leq \tau_0}|h_0(r)|, \sup_{0\leq r\leq \tau_0}|h'_0(r)|)$.

Set
$$f_n(u,r)=\exp\left({\frac{1}{2}\int_{0}^{r}\frac{\left(h_n(u,s)-\bar h_n(u,s)\right)^{2}}{s}ds}\right)\;.$$
Using~\eqref{hhbar0} and the induction hypothesis,
we can choose $\tau\leq \tau_0$ such that,
for all $(u,r)\in {\cal U}$,
\begin{equation}
\label{estF}
1\leq f_n(u,r)\leq  1+C\tau\;
\end{equation}
and consequently
\begin{equation}
\label{estdF}
|\partial_rf_n(u,r)|\leq C\tau\;,
\end{equation}
where $C=C(C_d)>0$. Then by decreasing $\tau$, if necessary, we find
\begin{equation}
\label{esttF}
1-C \tau\leq \tilde f_n(u,r)\leq 1+ C \tau\;,
\end{equation}
and
\begin{equation}
\label{estG}
|G_n(u,r)|\leq C\, r\;.
\end{equation}
So, from~\eqref{esttF} and \eqref{charEq} we conclude that
\begin{equation}
\label{estChar}
r_1+\frac{1-C\tau}{2}(u_1-u)\leq r_n(u;u_1,r_1) \leq  r_1+\frac{1+C\tau}{2}(u_1-u)\;.
\end{equation}
 A priori, the domain of definition of $h_{n+1}$  is composed of the points $(u_1,r_1)\in {\cal U}$ such that $\chi_n(u;u_1,r_1)\in {\cal U}$, for all $0\leq u\leq u_1$. But by choosing $\tau$ such that $1-C\tau>0$ and by setting $m>\frac{1+ C\tau}{2}$, we can use the previous estimates~\eqref{estChar} to make sure that this domain is in fact the entire $\cal U$.

From~\eqref{h_{n+1}_integral} and the induction hypothesis, together with~\eqref{hhbar2} and~\eqref{estG}, we find that
\begin{equation}
\label{rech}
 |h_{n+1}(u_1,r_1)|\leq \sup_{0\leq r \leq \tau_0}|h_0(r)| e^{C\tau^2}+\tau Ce^{C\tau^2}\;,
\end{equation}
and after setting $C_d\geq 2 \sup_{0\leq r \leq \tau_0}|h_0(r)|+1$, we conclude that there exists 
$$\tau=\tau\left(\sup_{0\leq r \leq \tau_0}|h_0(r)|, \sup_{0\leq r \leq \tau_0}|h'_0(r)|\right)$$ 
such that
\begin{equation}
\label{ind1}
\|h_{n+1}\|_{{\cal U}}\leq C_d\;.
\end{equation}
Moreover, since $h_n\in C^1(\cal U)$, we get that $\bar h_n, G_n$ and $\tilde f_n\in C^1({\cal U})$ and consequently, by differentiating~\eqref{h_{n+1}_integral}, $h_{n+1}\in C^1(\cal U)$. 

Now, the differential form of~\eqref{h_{n+1}_integral} is
$$
D_n h_{n+1}-G_n h_{n+1}=-G_n\bar h_n\;,
$$
which, when differentiated  with respect to $r$, leads to
\begin{equation}
\begin{aligned}
 D_{n}(\partial_{r}h_{n+1})-2G_{n}\partial_{r}h_{n+1}&=\partial_rG_{n}(h_{n+1}-\bar{h}_n)-G_n\partial_r\bar{h}_n
  \\
                                                     &=-J_{n}\partial_r\bar{h}_n-\left(J_{n}-G_{n}\right)\frac{(h_{n+1}-h_{n})}{r}\;,
\end{aligned}
\label{D_partial_h_{n+1}}
\end{equation}
where
\begin{equation}
\label{J}
J_n:=3G_n+3 f_n r+(3 r^{2}-1)\frac{1}{2}\partial_r f_n\;.
\end{equation}
From~\eqref{estG},~\eqref{estF} and~\eqref{estdF}, we have the estimate
\begin{equation}
\label{estJ}
|J_n(u,r)|\leq C\, r\;,
\end{equation}
which together with~\eqref{induction},~\eqref{estF},~\eqref{estdF} and~\eqref{estG} leads to
\begin{equation}
\label{drhn}
 \begin{aligned}
\left|\partial_r h_{n+1}(u_{1},r_{1})\right|
 &\leq
 \left|h'_{0}(r_n(0))\,e^{\int_{0}^{u_1}2 G_{n|_{\chi_{n}}}dv}\right|
 +
 \\
&\quad+
\int^{u_1}_{0}\left|J_{n}\partial_r\bar{h}_n+\left(J_{n}-G_n\right)\frac{(h_{n+1}-h_{n})}{r}\right|_{|_{\chi_{n}}}e^{\int^{u_1}_{u}2 G_{n|_{\chi_n}}dv}du
\\
&\leq \sup_{0\leq r \leq \tau_0}|h'_0(r)|e^{C\tau^2}+\tau^2 Ce^{C\tau^2}\;.
\end{aligned}
\end{equation}
As before, by increasing $C_d$ and decreasing $\tau$, if necessary, we obtain
\begin{equation}
\label{ind2}
\|\partial_r h_{n+1}\|_{{\cal U}}\leq C_d\;.
\end{equation}
We have just proved that, for all $n\in\mathbb{Z}^+$, $h_n\in C^1(\cal U)$ and 
\begin{equation}
\label{inductionProved}
 \|h_n\|_{\cal U}\leq C_d\;\;,\;\;\|\partial_rh_n\|_{\cal U}\leq C_d\;.
\end{equation}
Note, moreover, that given $(u_1,r_1)\in {\cal U}$ then $\chi_n(u;u_1,r_1)\in {\cal U}$, for all $0\leq u\leq u_1$ and all $n\in\mathbb{Z}^+$.

To show that the sequence $(h_n)$ is a contraction with respect to $\|\cdot\|_{\cal U}$ we consider the evolution equation
\begin{eqnarray}
\nonumber
D_{n}\left(h_{n+1}-h_{n}\right)
&=&
G_{n}\left(h_{n+1}-h_{n}\right)
+\left(G_{n}-G_{n-1}\right)\left(h_{n}-\bar h_{n}\right)
\\
\label{evolh-h}
&\quad&
-G_{n-1}\left(\bar h_{n}-\bar h_{n-1}\right)+
\frac{1}{2}\left(\tilde f_{n}-\tilde f_{n-1}\right)\partial_r h_{n}\;.
\end{eqnarray}
Noting that $(h_{n+1}-h_{n})(0,r)=0$, then the integration along $D_{n}$ gives
\begin{eqnarray}
\nonumber
\left|\left(h_{n+1}-h_{n}\right)(u_1,r_1)\right|
&\leq&
\int_{0}^{u_1}\left|\left(G_{n}-G_{n-1}\right)\left(h_{n}-\bar h_{n}\right)\right|_{|_{\chi_{n}}}e^{\int^{u_1}_{u} G_{n|_{\chi_{n}}}dv}du
\\
\nonumber
&\quad&+
\int_{0}^{u_1}\left|G_{n-1}\left(\bar h_{n}-\bar h_{n-1}\right)\right|_{|_{\chi_{n}}}e^{\int^{u_1}_{u} G_{n|_{\chi_{n}}}dv}du
\\
\nonumber
&\quad&+
\int_{0}^{u_1}\left|\frac{1}{2}\left(\tilde f_{n}-\tilde f_{n-1}\right)\partial_r h_{n}\right|_{|_{\chi_{n}}}e^{\int^{u_1}_{u} G_{n|_{\chi_{n}}}dv}du\;.
\label{deltahInt}
\end{eqnarray}
Using the fact that (compare with~\eqref{deltaf},~\eqref{deltaTildef} and~\eqref{deltaG})
\begin{equation}
\label{unifConvFG}
\|f_n-f_{n-1}\|_{\cal U} + \|\tilde f_n-\tilde f_{n-1}\|_{\cal U} + \|G_n-G_{n-1}\|_{\cal U}\leq C \|h_n-h_{n-1}\|_{\cal U}\;,
\end{equation}
together with the induction hypothesis~\eqref{induction} and~\eqref{estG}, we obtain
$$\|h_{n+1}-h_{n}\|_{\cal U}\leq C e^{C\tau}\tau \|h_n-h_{n-1}\|_{\cal U}\;.$$
By choosing $\tau$ small enough in the last estimate, we can make the constant on the right hand side smaller than unity.

In conclusion, we have established that there exists a choice of $\tau$, just depending on initial data, such that the sequence $(h_n)$
is a $\|\cdot\|_{\cal U}$-contraction.
As a consequence, $h_n$ converges uniformly to a continuous function $h$, in $[0,\tau]^2\subset \cal U$. It then follows immediately from~\eqref{unifConvFG} that 
$f_n$, $\tilde f_n$ and $G_n$ converge uniformly to $f=f[h]$, $\tilde f=\tilde f[h]$ and $G=G[h]$, respectively.
To study the uniform convergence of the sequence of characteristics $(\chi_{n})$, we start by using  equation~\eqref{Characteristic_ODE} to obtain
$$r_{n}(u;u_1,r_1)=r_1+\frac{1}{2}\int_u^{u_1}\tilde f_{n}(s,r_{n}(s;u_1,r_1))ds\;.$$
For $r(u)=r(u;u_1,r_1)$ and  $r_{n}(u)=r_{n}(u;u_1,r_1)$, we then have from \eqref{esttF}
\begin{equation}
\begin{aligned}
\label{deltaChi}
|r_{n}(u)-r(u)|
&\leq  \frac{1}{2}\int_u^{u_1}\left| \tilde f_{n}(s,r_n(s))-\tilde f(s,r(s))\right |ds
\\
&\leq
\frac{1}{2}\int_u^{u_1}\left|\tilde f_{n}(s,r_n(s))-\tilde f_{n}(s,r(s))\right |ds
+
\frac{1}{2}\int_u^{u_1}\left | \tilde f_{n}(s,r(s))-\tilde f(s,r(s))\right |ds\;
\\
&\leq
C \tau \int_u^{u_1}|r_{n}(s)-r(s)|ds\, ds
+
\frac{1}{2}\int_u^{u_1}\left| \tilde f_{n}(s,r(s))-\tilde f(s,r(s))\right | ds\;.
\end{aligned}
\end{equation}
 From the uniform convergence of $(\tilde f_{n})$, we find that for any $\varepsilon>0$ there exists $N\in\mathbb{Z}^+$ such that for
 $n\geq N$ we have
$$\frac{1}{2}\int_u^{u_1}\left|\tilde f_{n}(s,r(s))-\tilde f(s,r(s))\right| ds \leq\varepsilon\;.$$
We thus get that
$$|r_{n}(u)-r(u)|\leq \varepsilon+C\tau \int_u^{u_1}|r_{n}(s)-r(s)|ds\;$$
and Gr\"onwall's inequality gives
$$|r_{n}(u;u_1,r_1)-r(u;u_1,r_1)|\leq \varepsilon e^{C\tau(u_1-u)}\;,$$
so that the uniform convergence of each $(\chi_{n}(\cdot;u_1,r_1))$ to $\chi(\cdot;u_1,r_1)$, on $[0,\tau]^2$, follows.

So, we conclude that~\eqref{h_{n+1}_integral} converges uniformly to
\begin{equation}
 h(u_1,r_1)=h_{0}(\chi(0))e^{\int_0^{u_1}G|_{\chi}dv}-\int_0^{u_1}\left(G\bar{h}\right)|_{\chi}e^{\int^{u_1}_{u}G|_{\chi}dv}du\,,
\label{solutionIntegral}
\end{equation}
which is a continuous solution of~\eqref{mainIVPLocalUR} with continuous $Dh$.
To see that it is in fact a $C^1$ solution requires little more work. With that goal in mind we start with:
\begin{lem}
\label{lemEqui}
Under the conditions of Lemma~\ref{LemmaIteration}, the sequence $(\partial_rh_n)$ is equicontinuous.
\end{lem}
\begin{proof}
This proof is a detailed version of an argument given by Christodoulou in~\cite{Christodoulou:1986}. Start by noting that the sequence $(h_n)$ is equicontinuous: In fact, this follows from the fact that both sequences $(\partial_r h_{n})$ and $(D_{n-1}h_{n})$ are equibounded as a consequence of~\eqref{inductionProved}. Then, it follows that all sequences appearing in~\eqref{D_partial_h_{n+1}} are equicontinuous as well.

We will now show that $(\partial_rh_n)$ is equicontinuous with respect to $r$. In order to do that, let $r_2>r_1\ge 0$ and define
$$\psi_{n+1}(u)=\partial_rh_{n+1}\circ \chi_n(u;u_1,r_2)-\partial_rh_{n+1}\circ \chi_n(u;u_1,r_1)\;.$$
Differentiating with respect to $u$, we get from~\eqref{D_partial_h_{n+1}}
\begin{eqnarray}
\nonumber
\psi_{n+1}'(u)
&=&
D_n\partial_rh_{n+1}\circ \chi_n(u;u_1,r_2)-D_n\partial_rh_{n+1}\circ \chi_n(u;u_1,r_1)
\\
\nonumber
&=&
2(G_n\partial_rh_{n+1})\circ \chi_n(u;u_1,r_2)-2(G_n\partial_rh_{n+1})\circ \chi_n(u;u_1,r_1)
\\
\nonumber
&\quad&
+F_n\circ \chi_n(u;u_1,r_2)-F_n\circ \chi_n(u;u_1,r_1)
\\
\label{psi'}
&=&
\left(2G_n\circ \chi_n(u;u_1,r_2)\right)\,\psi_{n+1}(u)+\hat F_n(u)\,,
\end{eqnarray}
where
$$F_n:=-J_{n}\partial_r\bar{h}_{n}-\left(J_{n}-G_{n}\right)\frac{(h_{n+1}-h_{n})}{r}\;,$$
and
\begin{eqnarray*}
	\hat F_n(u)
	&=&
	F_n(u)\circ \chi_n(u;u_1,r_2)-F_n\circ \chi_n(u;u_1,r_1)
	\\
	&\quad&
	+2\left[G_n\circ \chi_n(u;u_1,r_2)-G_n\circ \chi_n(u;u_1,r_1)\right][\partial_rh_{n+1}\circ \chi_n(u;u_1,r_1)]\;.
\end{eqnarray*}
Then, integrating~\eqref{psi'} and using~\eqref{estG} we see that
\begin{equation}\label{intPsi}
|\psi_{n+1}(u)|\leq \tilde C |\psi_{n+1}(0)|+ \tilde C \int_0^u |\hat F_n(u)|du\;,
\end{equation}
with $\tilde C>0$ constant. 

Now let $\varepsilon>0$. Since $\hat F_n$ is constructed out of the sum, product and composition of equicontinuous sequences and the equibounded sequence $(\partial_rh_{n})$, we see that there exists $\varepsilon_0$ sufficiently small and independent of $n$ such that, for $|r_2-r_1|<\varepsilon_0$, we have
$$\hat F_n(u)\leq \frac{\varepsilon}{2\tau \tilde C}\;.$$
The regularity of the initial data allows us to conclude that, by decreasing $\varepsilon_0$ if necessary, we have
$$|\psi_{n+1}(0)|=|\partial_rh_0(0,r_2)-\partial_rh_0(0,r_1)|<\frac{\varepsilon}{2 \tilde C}\;.$$
Then,~\eqref{intPsi} gives
$$|\psi_{n+1}(u)|\leq \varepsilon\;, $$
and the equicontinuity of $(\partial_rh_n)$ with respect to $r$ follows.

To finish, one just needs to observe that since $(D_{n+1}\partial_rh_n)$ is equibounded according to the proof of~\eqref{inductionProved}, then $(\partial_rh_n)$ is also equicontinuous with respect to $u$.
\end{proof}
Now, arguing as in~\cite{Christodoulou:1986}, since the sequence $(\partial_rh_n)$ is equicontinuous, then the Arzel\`a-Ascoli Theorem guarantees the existence of a subsequence, that we also denote by $(\partial_r h_{n})$, converging uniformly, on $[0,\tau]^2$, to a continuous function $w$. If we then consider
$$\hat h(u,r)= h(u,0)+\int_0^r w(u,s)\,ds\;,$$
and write
$$h_n(u,r)= h(u,0)+\int_0^r \partial_r h_n(u,s)\,ds\;,$$
it becomes clear that $\|h_n-\hat h\|_{\cal U}\rightarrow 0$, as $n\rightarrow \infty$, from which we see that $\hat h= h$ and therefore $\partial_r  h=\partial_r \hat h =w$. We have thus concluded that $h$ is a $C^1$ solution of~\eqref{mainIVPLocalUR}.

We can show that the solution is as regular as the initial data by differentiating the integral version of~\eqref{mainEq}  (compare with~\cite[(17)]{CostaSpherically}). The proof of uniqueness is similar, although simpler, to the proof of Proposition~\ref{propUniq} (compare with~\cite[(59)]{CostaProblem}). 

\end{proof}

\section*{Acknowledgments}
The authors thank Rodrigo Duarte for his comments and careful reading of the paper, the Erwin Schr\"odinger International Institute for Mathematical Physics, ESI, where part of this work has been done, FCT project PTDC/MAT-ANA/1275/2014 and CAMGSD, IST, Univ. Lisboa, through FCT project UID/MAT/04459/2019. FCM also thanks CMAT, Univ. Minho, through FCT project Est-OE/MAT/UIDB/00013/2020 and FEDER Funds COMPETE.

\end{document}